\documentclass[reprint,amsmath,amssymb,aps,pra,floatfix,superscriptaddress]{revtex4-2}

\usepackage{blindtext}
\usepackage{lipsum}
\usepackage{physics}
\usepackage{centernot}
\usepackage{wrapfig}
\usepackage{dsfont}
\usepackage{comment}
\usepackage{amsthm}
\usepackage{xfrac}
\usepackage{xcolor}
\usepackage{graphicx}
\usepackage{enumitem}
\usepackage{tgtermes}
\usepackage{orcidlink}
\usepackage{hyperref}
\usepackage{xcolor}
\usepackage{appendix}
\hypersetup{colorlinks=true, allcolors=blue}

\providecommand*{\unit}[1]{\,\ifmmode
\mathrm{\,#1}\else\textup{#1}\fi}

\newcommand{\daga}[1]{{#1}^{\dagger}}
\newcommand{\normt}[1]{\norm{#1}_1}
\newcommand{\Md}[1]{M_{#1}(\mathbb{C})}
\newcommand{\family}[1]{\{{#1}\}_{t\ge0}}

\newcommand{\wt}[1]{\widetilde{#1}}
\newcommand{\mc}[1]{\mathcal{#1}}
\newcommand{\bs}[1]{\boldsymbol{#1}}

\newtheorem{theorem}{Theorem}

\newtheorem{proposition}{Proposition}
\newtheorem{lemma}{Lemma}
\newtheoremstyle{boldremark}
{\dimexpr\topsep/2\relax}
{\dimexpr\topsep/2\relax}
{}          
{}          
{\bfseries}
{.}
{.5em}
{}

\newtheorem*{theorem*}{Theorem}
\newtheorem{remark}{Remark}
\theoremstyle{plain}
\newtheorem{example}{Example}
\theoremstyle{plain}
\newtheorem{corollary}{Corollary}

\begin{document}

\title{Quantum vs. classical $P$-divisibility}

\author{Fabio Benatti}
\email{benatti@ts.infn.it}
\affiliation{Dipartimento di Fisica, Universit\`{a} di Trieste, 34151 Trieste, Italy}
\affiliation{Istituto Nazionale di Fisica Nucleare, Sezione di Trieste, 34151 Trieste, Italy}

\author{Dariusz Chru\ifmmode \acute{s}\else \'{s}\fi{}ci\ifmmode \acute{n}\else \'{n}\fi{}ski}
\affiliation{Institute of Physics, Faculty of Physics, Astronomy and Informatics, Nicolaus Copernicus University, Grudziadzka 5/7, 87-100 Toru\'{n}, Poland}

\author{Giovanni Nichele}
\affiliation{Dipartimento di Fisica, Universit\`{a} di Trieste, 34151 Trieste, Italy}
\affiliation{Istituto Nazionale di Fisica Nucleare, Sezione di Trieste, 34151 Trieste, Italy}

\begin{abstract}
$P$-divisibility is a central concept in both classical and quantum non-Markovian processes; in particular, it is strictly related to the notion of information backflow.
When restricted to a fixed commutative algebra generated by a complete set of orthogonal projections, any quantum dynamics naturally provides a classical stochastic process. It is indeed well known that a quantum generator gives rise to a $P$-divisible quantum dynamics if and only if all its possible classical reductions give rise to divisible classical stochastic processes. Yet, this property does not hold if one operates a classical reduction of the quantum dynamical maps instead of their generators:  as an  example, for a unitary dynamics, $P$-divisibility of its classical reduction is inevitably lost and the latter thus exhibits information backflow. Instead, for some important classes of purely dissipative qubit evolutions, quantum $P$-divisibility always implies classical $P$-divisibility and  therefore
excludes information backflow both in the quantum
and classical scenarios. On the contrary, for a wide class of orthogonally covariant qubit dynamics, we show that loss of classical $P$-divisibility originates from the classical reduction of a purely dissipative $P$-divisible quantum dynamics as in the unitary case. Moreover, such an effect can be interpreted in terms of information backflow due to the coherences developed by the quantumly evolving classical state.
\end{abstract}

\maketitle

\section{Introduction}

In the following, we shall be concerned with $n$-level open quantum systems interacting with their environment in a way that cannot be approximated by
means of dynamical semigroups~\cite{AlickiLendi,BreuerPetruccione,RivasHuelga}. The presence of memory effects is embodied
by one-parameter families of completely positive and trace-preserving ($CPTP$) maps
$\{\Lambda_t\}_{t\geq 0}$ that act on the convex space $\mathcal{S}(M_n(\mathbb{C}))$ of their positive and unit trace density matrices  $\rho\in M_n(\mathbb{C})$ such that $\Lambda_{t+s}\neq \Lambda_t\Lambda_s$.
We shall assume the dynamics $\rho\mapsto\rho(t)=\Lambda_t[\rho]$ to consist of algebraically invertible maps generated by  time-dependent generators $\mathcal{L}_t$ on $\mathcal{S}(M_n(\mathbb{C}))$, namely by master equations of the form
\begin{equation}
	\label{eq0}
	\dot\Lambda_t=\mathcal{L}_t\,\Lambda_t\ , \quad \mathcal{L}_t:=\dot\Lambda_t\Lambda_t^{-1}\ ,
\end{equation}
which are solved by the time-ordered exponentials
\begin{align}
	\Lambda_t&=\mathcal{T}_\leftarrow e^{\int_0^t\dd{s}\,\mathcal{L}_s} \nonumber\\
    &=\sum_{k=0}^\infty\int_0^t\dd s_1\int_0^{s_2}\hskip-.1cm\dd s_2\cdots\int_0^{s_{k-1}}\hskip -.5cm\dd s_k\,\mathcal{L}_{s_1}\cdots\mathcal{L}_{s_k}\ .\label{timeordering}
\end{align}
Because of the invertibility of the dynamics, the intertwining maps $\Lambda_{t,s}$ connecting time $s$ to time $t\geq s\geq 0$,
$\Lambda_t=\Lambda_{t,s}\Lambda_s$, are given by
\begin{equation}
	\label{intertwiners}
\Lambda_{t,s}:=\Lambda_t^{\phantom{.}}\,\Lambda_s^{-1}\ ,\quad t\geq s\geq 0\ .
\end{equation}
Given $n^2$ Hilbert-Schmidt orthonormal matrices $F_k\in M_n(\mathbb{C})$ such that ${\rm Tr}(F^\dag_jF_k)=\delta_{jk}$,
with $F_0=\mathds{1}/\sqrt{n}$, the generator can always be written as
\begin{equation}
	\label{GKSL}
	\mathcal{L}_t[\rho]=-i[H(t),\rho]+\mathcal{D}_t[\rho]\ ,
\end{equation}
with $H(t)$ being a time-dependent Hamiltonian and dissipator
\begin{equation}\label{dissipator}
	\mathcal{D}_t[\rho]=\sum_{i,j=1}^{n^2-1} K_{ij}(t)\left(F_i\rho F^\dag_j-\frac{1}{2}\{F^\dag_jF_i,\rho\}\right)\,,
\end{equation}
with Hermitian, time-dependent $(n^2-1)\times (n^2-1)$ Kossakowski matrix $K(t)=[K_{ij}(t)]=K^\dag(t)$. According to the GKSL-theorem~\cite{GoriniKossSud,LindbladTh}, for time-independent generators, such that $\dot \Lambda_t=\mathcal{L}\Lambda_t$, the dynamics becomes a semi-group,
\begin{equation}
	\label{quantumsemigroup}
	\Lambda_t={\rm e}^{t\mathcal{L}}\ ,\quad \Lambda_{t+s}=\Lambda_{t}\Lambda_s, \qquad \Lambda_{t=0}=\mathrm{id}\quad \forall t,s\ge0\ ,
\end{equation}
and the complete positivity of the maps $\Lambda_t$ is equivalent to the positive semi-definiteness of the time-independent Kossakowski matrix:
$K(t)=K=[K_{jk}]\geq 0$.
Instead, in the time-dependent case there are no general constraints on $K(t)$  characterizing the complete positivity of the generated $\Lambda_t$.
What certainly holds is that $K(t)\geq 0$ is equivalent to the complete positivity of all the intertwiners $\Lambda_{t,s}$ \cite{RivasHuelga,ChrusReview22} (and thus implies the complete positivity  of $\Lambda_t$). In this case, the dynamics is called $CP$-divisible.
However, the complete positivity of $\Lambda_t$ does not require $K(t)\geq 0$; indeed, completely positive $\Lambda_t$ can very well be only $P$-divisible, namely with intertwiners $\Lambda_{t,s}$ that are only positive and not completely positive.
Even when not semi-groups, $P$-divisible families of $CPTP$ maps may be considered to be Markovian as they do not give rise to the phenomenon of backflow of information, which we will be referring to as $BFI$ in the following,  from the environment to the open quantum system immersed in it~\cite{BLP,ColloquiumBreuer,FBGN2023}.
Indeed, Markovianity is nowadays identified with information being lost from an open system to its environment, leading to increasing lack of distinguishability of generic pairs of open system states in the course of time.

The classical analogue of a  one-parameter family of quantum dynamical maps
is a one-parameter family of $n\times n$ stochastic matrices $\{T(t)\}_{t\ge0}$, with entries $T_{ij}(t)$ such that $T_{ij}(t)\ge0$, $\sum_i T_{ij}(t)=1$.
These matrices transform probability vectors $\bs{p}=\{p_i\}_{i=1}^n$, $p_i\geq 0$, $\sum_{j=1}^np_j=1$,  into themselves, $\bs{p}(t)=T(t)\bs{p}$.

In the following we shall study the properties inherited by the classical dynamics that naturally emerges from $P$-divisible
quantum dynamics. The classical dynamics will be obtained by restricting both the quantum dynamics and its tomographic reconstruction to a same (maximally) Abelian sub-algebra of $M_n(\mathbb{C})$ generated by any given choice of  $n$ orthogonal rank-$1$ projections $P_j\in M_n(\mathbb{C})$ such that
$P_jP_k=\delta_{jk}$ and $\sum_{j=1}^nP_j=\mathds{1}$.  In the following, for sake of simplicity, we shall identify the commutative sub-algebra $\mathcal{P}$ with the orthonormal, complete set $\{P_i\}_{i=1}^n$ that generates it and refer to the quantum dynamics restricted to it as its \textit{classical reduction}.

In order to extract a classical dynamics from a quantum one, one can follow two prescriptions;  one can either reduce the generator of the quantum evolution, or reduce the quantum dynamics itself. The latter reduction procedure is somewhat  dual to the so-called embeddability
of a given classical stochastic matrix into a Markovian quantum evolution, that may give rise to some quantum advantage~\cite{LostaglioKorzekwa2021quantum,shahbeigi2023quantum}.
The two different points of view yield different classical dynamics that are discussed in Section \ref{sec:Comparison}.
In particular, regarding the property of $P$-divisibility, the following question naturally emerges: \hfill\break
\textit{Given a $P$-divisible dynamics $\family{\Lambda_t}$, are their classical reductions $P$-divisible?}  \par
As we shall see, by reducing the generator, the answer to the previous question \color{black}
is somewhat trivial, in that quantum $P$-divisibility is equivalent to classical $P$-divisibility. Conversely,
by reducing the dynamical map one
 has the peculiar effect that a $P$-divisible quantum dynamics need not give rise to classically reduced divisible dynamics.

Such an issue is not only mathematically interesting; indeed, as we shall see, the loss of $P$-divisibility by the classical reduction~\eqref{classicaldynFROMGEN-t} can be interpreted from a physical perspective, in terms of backflow of information.

In the same Section~\ref{sec:Comparison}, after presenting the general framework of classical reductions, we will illustrate the main results of the work which will be proved in Sections~\ref{sec:classicalfromquantum}-\ref{sec:class}.

\section{Classical reductions and results}
\label{sec:Comparison}

As explained in the Introduction, 
upon choosing a maximally Abelian subalgebra $\mathcal{P}$ of $\Md{n}$ one may proceed along two different paths.
Namely, given the generator $\mathcal{L}_t$ of the quantum dynamics $\Lambda_t$, one can consider the  $n\times n$ matrix $L(t)$ with entries
\begin{equation}
		\label{classicaldynFROMGEN-t}
		L_{ij}(t)=\Tr(P_i\mathcal{L}_t[P_j])
	\end{equation}
and study the one-parameter family of $n\times n$ stochastic matrices $D(t)$ solving the classical master equation
\begin{equation}
\label{class-master}
\dot{D}(t)=L(t)\,D(t)\,.
\end{equation}
Another possibility is to focus upon the classical reduction of the generated quantum dynamics  $\Lambda_t$ itself; namely, to restrict to the stochastic matrix $T(t)$ with entries
\begin{equation}
		\label{classicaldyn-t}
		T_{ij}(t)=\Tr(P_i\Lambda_t[P_j])]\ .
	\end{equation}
The physical interpretation of $T(t)$ becomes clear by introducing the decohering  map
\begin{equation}
\label{decoheringmap1}
\mathbb{P}[\rho]=\sum_{i=1}^n P_i \rho P_i
\end{equation}
which diagonalizes every state $\rho$ along the orthonormal basis associated with the chosen complete set of orthonormal projections
$\mathcal{P}=\{P_i\}_{i=1}^n$.
Then, the map $\Lambda^{\mathbb{P}}_t:=\mathbb{P}\,\Lambda_t\,\mathbb{P}$ acts initially on diagonal matrices, turns them into non-diagonal ones at time $t$ and the developed coherences are then suppressed by the left-most $\mathbb{P}$.
Explicitly, for all states $\rho$, one gets:
\begin{equation}
\label{decoheringmap2}
\Lambda^{\mathbb{P}}_t[\rho]=\sum_{j,k=1}^n T_{jk}(t)\, \Tr(P_k\rho)\,P_j\ .
\end{equation}
Practically speaking, in this approach, at each instant of time, the quantum dynamics is preceded and followed by two projective measurements specified by
$\mathcal{P}$.

To compare the latter classical dynamics with the one obtained from~\eqref{classicaldynFROMGEN-t}, for sake of simplicity we restrict to the case when
the quantum dynamics $\Lambda_t$ is a semigroup with time-independent generator $\mathcal{L}$: $\Lambda_t=e^{t\mathcal{L}}$.
It thus follows that $L(t)=L$ in~\eqref{class-master} and $D(t)=e^{t\,L}$ so that:
\begin{eqnarray}
\label{decoheringmmap3}
\mathcal{L}^\mathbb{P}[\rho]&:=&\mathbb{P}\,\mathcal{L}\,\mathbb{P}[\rho]=\sum_{j,k=1}^n L_{jk}\, \Tr(P_k\,\rho)\,P_j\ ,\\
\label{decoheringmmap4}
e^{t\mathcal{L}^\mathbb{P}}\mathbb{P}[\rho]&=&\sum_{j,k=1}^n D_{jk}(t)\, \Tr(P_k\,\rho)\,P_j\ .
\end{eqnarray}
On the other hand, fix $t\geq 0$ and consider the large $N$ behaviour
\begin{eqnarray}
\label{decoheringmmap5}
\left(\mathbb{P}\,\Lambda_{t/N}\,\mathbb{P}\right)^N\simeq \left(\mathds{1}+\frac{t}{N}\,\mathcal{L}^\mathbb{P}\right)^N\mathbb{P}\simeq
e^{t\mathcal{L}^\mathbb{P}}\mathbb{P}\ .
\end{eqnarray}
Therefore, unlike when reducing the quantum evolution up to time $t$ by an initial and a final projective measurements, the classical reduction of the quantum generator physically amounts to letting the quantum dynamics to act  in between iterated projective measurements
separated by smaller and smaller time intervals.

\begin{remark}
\label{rem:semigroup}
Notice that
\begin{eqnarray*}
\left(\mathbb{P}\,\Lambda_\epsilon\,\mathbb{P}\right)^N[\rho]&=&\sum_{i,j_N,\ldots,j_1}T_{ij_N}(\epsilon)\cdots T_{j_2j_1}(\epsilon)\, \Tr(\rho P_{j_1})\, P_i\\
&=&\sum_{i,j}\left[\big(T(\epsilon)\big)^N\right]_{ij}\, \Tr(\rho P_j)\, P_i\ ,
\end{eqnarray*}
where $\epsilon=t/N$. On the other hand, in the semi-group case, from~\eqref{classicaldynFROMGEN-t} and~\eqref{classicaldyn-t}, with $\mathcal{L}_t=\mathcal{L}$ and $L(t)=L$,  it follows that $T(\epsilon)\simeq \mathds{1}+\epsilon\, L\simeq D(\epsilon)$; therefore,
\begin{equation*}
\lim_{N\to+\infty} T(t/N)^N=D(t)\neq T(t)\ .
\end{equation*}
The origin of the discrepancy is the fact that the family of matrices $T(t)$ is not a semigroup; indeed, if $T(t)$ is forced to be a semigroup, so that
$T(t+\epsilon)=T(\epsilon)T(t)$, then one has
$T(t)=D(t)$.
\end{remark}

\subsection{Results}
\label{sec:Results}

We will mainly focus not upon classical reductions of generators, rather upon those of the generated quantum dynamics and illustrate
a number of results in this context. They can be summarized as follows.

First results, some of them already present in the literature, regard the connections between $P$-divisibility and classical and quantum stochastic processes. They are fundamental tools for what follows and are reported in Section~\ref{sec:classicalfromquantum}.

Other results concern an interesting  physical scenario that involves  a $P$-divisible quantum dynamics whose classical reduction is not $P$-divisible.
Indeed, the quantum dynamics does not exhibit $BFI$, while the classical one does.
In Section~\ref{sec:BFI_interpretation}, by working in the so-called $BLP$ framework~\cite{BLP,ColloquiumBreuer},
we show that classical $BFI$ at a given instant of time  is due to information being stored in the coherences of the quantum states at previous times.
These coherences are built up by the quantum dynamics as it generically maps out of the chosen commutative algebra used for the reduction.

Further results regard unital qubit dynamics, for which $P$-divisibility of both the quantum dynamics and its classical reduction are to some extent under control. In Section \ref{sec:unitalqubit}, we show that loss of $P$-divisibility for the classical reduction 
is typical of unitary evolutions, but not of purely dissipative (namely, whose generator is such that $\mathcal{L}_t=\mathcal{D}_t$, see~\eqref{GKSL}) and self-dual qubit dynamics. Indeed, for them we show that  $P$-divisibility of the reduced classical stochastic processes is equivalent to that of the quantum map that originates it.

The last results concern a wider class of of orthogonally covariant qubit dynamics; in Section~\ref{sec:class}, using necessary and sufficient conditions for $P$ and $CP$-divisibility in terms of the generators, we provide  a fairly general construction of purely dissipative quantum dynamics whose classical reductions are not $P$-divisible.

\section{Classical reductions and $P$-divisibility}
\label{sec:classicalfromquantum}

In this Section, we discuss some fundamental preliminary results about classical reductions, first considering the case of time-independent positive maps and time independent-generators of positive semigroups before moving to the time-inhomogeneous case of $P$-divisible families.
\color{black}

\subsection{Classical Reduction of $PTP$ maps}

As much as for one parameter families of dynamical maps in~\eqref{classicaldyn-t}, given an orthonormal complete set of orthogonal rank-1 projections $\{P_i\}_{i=1}^n$ spanning a maximally Abelian algebra $\mathcal{P}$, any $PTP$ quantum map $\Phi$ on $\mathcal{S}(M_n(\mathbb{C}))$  defines a stochastic matrix through the entries
\begin{equation}\label{stochasticfromquantum}
	T_{ij}:=\Tr(P_i\Phi[P_j])\,,\qquad P_i P_j=\delta_{ij} P_i\,.
\end{equation}
Indeed, the positivity of $\Phi$ together with the completeness of the projections and trace-preservation yields $T_{ij}\geq 0$ and $\sum_{i=1}^nT_{ij}=1$. The converse, though, is not true: even if the matrices $T$ defined by $\Phi$ through~\eqref{stochasticfromquantum} are stochastic for
all maximally Abelian $\mathcal{P}$, $\Phi$ need not be positive.
The reason is that the positivity of $\Phi$ amounts to
\begin{equation}
	\label{positivity}
	\Tr(P\Phi[Q])\ge0
\end{equation}
for all one-dimensional, not necessarily orthogonal projections $P$ and $Q$, while \eqref{stochasticfromquantum} restricts to orthogonal ones.

\begin{example}
	\label{example:classstochdoesntimplypos}
	Let $\Phi$ be a Hermiticity and trace-preserving, unital qubit map on $\mathcal{S}(\Md{2})$, namely $\Phi[\sigma_0]=\sigma_0$. Here $\sigma_0=\mathds{1}$ denotes the identity $2\times 2$ matrix which, together with the other Pauli matrices, $\sigma_{1,2,3}$, after normalization, constitutes a Hilbert-Schmidt basis:
	$	\displaystyle{\rm Tr}\left(\frac{\sigma_\alpha}{\sqrt{2}}\,\frac{\sigma_\beta}{\sqrt{2}}\right)=\delta_{\alpha\beta}$.
	They	provide a matrix representation of $\Phi$ through the entries
	\begin{equation}
		\label{Paulirep}
		\Phi_{\alpha\beta}:=\frac{1}{2}\Tr\left({\sigma_\alpha}\,\Phi\left[{\sigma_\beta}\right]\right)\ , \quad
		\alpha=0,\dots,3\ .
	\end{equation}
	where  $\Phi_{00}=1$ and $\Phi_{0i}=0$ for $i=1,2,3$ follows from trace preservation. Also, from unitality it follows that $\Phi_{i0}=0$; therefore, one concentrates upon the real matrix $\widetilde{\Phi}$ with entries
	\begin{equation}
		\widetilde\Phi_{ij}:=\frac{1}{2}\Tr(\sigma_i\Phi[\sigma_j]), \quad i=1,2,3\,,
	\end{equation}
	so that	$\Phi[\sigma_i]=\sum_{j=1}^3 \widetilde{\Phi}_{ji}\, \sigma_j, $.
	Let $P_0$, $P_1=\mathds{1}-P_0$ arbitrary orthonormal projectors in $M_2(\mathbb{C})$ and
	define the $2\times 2$ matrix with entries $T_{ij}=\Tr(P_i\Phi[P_j])$. Since $\Phi$ is unital,
	\begin{equation}\label{tindepstochasticmatrix}
		T=\begin{pmatrix}
			T_{00} & 1-T_{00}\\
			1-T_{00} & T_{00}
		\end{pmatrix},
	\end{equation}
	is a (bi-)stochastic matrix iff $0\le T_{00}\le1$.  In the Bloch representation, $P_0$  is identified by means of a unit real vector $\boldsymbol{n}\in\mathbb{R}^3$, $P_{\boldsymbol{n}}=(\mathds{1}+\boldsymbol{n}\cdot\boldsymbol\sigma)/2$, where $\boldsymbol{\sigma}=(\sigma_1,\sigma_2,\sigma_3)$.
	Then, the bistochasticity conditions read
	\begin{align*}
		1\geq T_{00}=\Tr(P_0\Phi[P_0])=\frac{1}{2}+\frac{1}{2}\mel{\boldsymbol{n}}{\widetilde{\Phi}}{\boldsymbol{n}}\ge0,
	\end{align*}
	for all unit $\boldsymbol{n}\in \mathbb{R}^3$. Hence, $T$ is stochastic for every choice of orthogonal projectors $P_0$ and $\mathds{1}-P_0$ iff
	\begin{equation}
		\ell(\widetilde\Phi):=\sup_{\norm{\boldsymbol{n}}=1}|\mel{\boldsymbol{n}}{\widetilde{\Phi}}{\boldsymbol{n}}|\le1.
	\end{equation}
	Notice that, if $\widetilde\Phi$ is non-symmetric, $\widetilde{\Phi}\neq \widetilde{\Phi}^T$, where ${}^T$ denotes transposition,   $\ell(\widetilde\Phi)$ is in general smaller than the norm
	\begin{equation}
		\label{norm}
		\|\widetilde{\Phi}\|=\sup_{\substack{\boldsymbol{n},\boldsymbol{m}\\\norm{\boldsymbol{n}}=\norm{\boldsymbol{m}}=1}} |\mel{\boldsymbol{n}}{\widetilde{\Phi}}{\boldsymbol{m}}|\ .
	\end{equation}
	On the other hand, the map is positive iff
	\begin{equation}
		\label{positivity_rewritten}
		\Tr(P_{\pm\boldsymbol{n}}\Phi[P_{\vb{m}}])=\frac{1}{2}\pm\frac{1}{2}\mel{\boldsymbol{n}}{\widetilde{\Phi}}{\boldsymbol{m}}\ge0\,,
	\end{equation}
	for all unit vectors $\boldsymbol{n},\boldsymbol{m}\in\mathbb{R}^3$; namely, iff $\vert\mel{\boldsymbol{n}}{\widetilde{\Phi}}{\boldsymbol{m}}\vert\leq 1$. Then, $\|\widetilde{\Phi}\|\leq 1$, since $\wt{\Phi}$ is real.
	If $\widetilde{\Phi}=\widetilde{\Phi}^T$, then $\|{\widetilde{\Phi}}\|=\ell(\widetilde{\Phi})$ and $T$ is stochastic for all ${P}_0$ iff the map $\Phi$ is positive.
	
	Instead, let $\alpha\in (0,\frac{\pi}{2})$, $\lambda>\eta\ge\lambda \cos(\alpha)>0$; then,
	\begin{equation}
		\widetilde{\Phi}=\begin{pmatrix}
			\lambda \cos(\alpha) & -\lambda \sin(\alpha) & 0 \\
			\lambda \sin(\alpha) &	\lambda \cos(\alpha) & 0 \\
			0					 &  0					 & \eta \\
		\end{pmatrix}\neq \widetilde{\Phi}^T
	\end{equation}
	has singular values $\lambda$ and 	$\eta$. Therefore, $\|{\widetilde{\Phi}}\|=\lambda$, while $|\mel{\boldsymbol{n}}{\widetilde{\Phi}}{\boldsymbol{n}}|=\lambda \cos(\alpha)+n_3^2\,(\eta -\lambda\cos(\alpha))$.
	Then, choosing $\lambda=1+\varepsilon$ and $\eta=\lambda\cos(\alpha)$, one has
	\begin{align*}
		\ell(\widetilde{\Phi})=\sup_{\norm{\boldsymbol{n}}=1}|\mel{\boldsymbol{n}}{\widetilde{\Phi}}{\boldsymbol{n}}|= (1+\varepsilon)\cos(\alpha)\le1<\|\widetilde{\Phi}\|.
	\end{align*}
	for $\varepsilon$ sufficiently small. With such choice of parameters, the map ${\Phi}$ is not positive, while $T_{ij}=\Tr(P_i\Phi[P_j])$  are the entries of a well defined stochastic matrix, for all choices of orthonormal projectors $P_0$ and $P_1$.
\end{example}

\subsection{Classical from quantum semi-groups}
\label{qcsem:sec}

For semigroups of quantum maps as in~\eqref{quantumsemigroup}, trace preservation imposes  $\Tr(\mathcal{L}[\rho])=0$ for all $\rho\in\mathcal{S}(M_n(\mathbb{C}))$. Kossakowski \cite{Kossakowski} found necessary and sufficient conditions for $\mathcal{L}$ to be the generator of a positive trace-preserving ($PTP$) semigroup $\{\Lambda_t\}_{t\ge0}$.
\begin{theorem}
	\label{kossakowskiconditions}
	A linear map $\mathcal{L}$ on $\mathcal{S}(M_n(\mathbb{C}))$  generates a semigroup of PTP maps $\family{\Lambda_t}$ iff
	\begin{equation}
		\Tr(Q\mathcal{L}[P])\ge0\ ,\quad \Tr(\mathcal{L}[\mathds{\rho}])=0
	\end{equation}
	for all pairs of projectors $P\perp Q$ and all $\rho\in\mathcal{S}(M_n(\mathbb{C}))$.
\end{theorem}

A semigroup of $n\times n$ stochastic matrices $T(t)=e^{t\,  L}$ with generating matrix $L=[L_{ij}]$ naturally gives rise to a semigroup of operators on every commutative
algebra $\mathcal{P}$ spanned by $\{P_i\}_{i=1}^n$ orthogonal projectors such that $\sum_{i=1}^nP_i=\mathds{1}$. Indeed, using the decohering map introduced in~\eqref{decoheringmap1}, one sets
	\begin{align*}
		\mathcal{L}^{\mathbb{P}}[\rho]&=\sum_{ij} L_{ij} \Tr(\rho P_j) \, P_i\,,
	\end{align*}
and exponentiates obtaining a semigroup of maps $\Lambda^{\mathbb{P}}_t={\rm e}^{t\mathcal{L}^{\mathbb{P}}}$ on the algebra of $n\times n$ matrices.
It turns out that 	
$$
\Lambda_t^{\mathbb{P}}[\rho]=\sum_{ij} T_{ij}(t) \Tr(\rho P_j) P_i\ .
$$
Then, reducing $\mathcal{L}^{\mathbb{P}}$ or $\Lambda_t^{\mathbb{P}}$ to the algebra $\mathcal{P}$, one ends up with the same classical  process
$T(t)=e^{t L}$.
Thus, from Theorem~\ref{kossakowskiconditions}
	\begin{equation}
	\begin{aligned}
			L_{ij}=\Tr(P_i\mathcal{L}^{\mathbb{P}}[P_j])\ge0 \quad i\ne j\, \\ \qquad \sum_i\Tr(P_i\mathcal{L}^{\mathbb{P}}[P_j])=\sum_{i=1}^n L_{ij}=0\ . \label{kolmorovconditions}
	\end{aligned}
	\end{equation}
	which are the so-called Kolmogorov conditions~\cite{ChrusReview22,VanKampen}, that are necessary and sufficient for an $n\times n$ matrix $L$ to generate a semigroup of stochastic matrices ${T(t)={e}^{tL}}$. This can be seen by expanding the semigroup $T(\epsilon)$ for $\epsilon\ll1$,
	yielding
	\begin{equation}\label{Kolmogorv_expansion}
	T_{ij}(\epsilon)=\delta_{ij}+ \epsilon\, L_{ij}  + O(\epsilon^2)\,,
	\end{equation}
	Necessity of \eqref{kolmorovconditions} is straightforwardly deduced from~\eqref{Kolmogorv_expansion}, while sufficiency follows by invoking the semigroup property.
	
 In conjunction with the previous observation, one derives the following natural consequence from Theorem \ref{kossakowskiconditions}.
\begin{corollary}\label{corollary1}
	A linear map $\mathcal{L}$ on $\mathcal{S}(M_n(\mathbb{C}))$  such that  $\Tr(\mathcal{L}[\mathds{\rho}])=0$, generates a semigroup of $PTP$
	maps $\family{\Lambda_t}$ iff the $n\times n$ matrix $L$ with entries
	\begin{equation}
		\label{classicaldynFROMGEN}
		L_{ij}=\Tr(P_i\mathcal{L}[P_j])
	\end{equation}
	generates a semigroup of stochastic $n\times n$ matrices, for any orthonormal, complete set $\{P_i\}_{i=1}^n$ of rank-1 projectors.
\end{corollary}
\begin{proof}
	The ``only if'' part is trivial. For the ``if'' part, suppose that $L$ in~\eqref{classicaldynFROMGEN} generates a positive stochastic semigroup for all choices of $\{P_i\}_i$. This is equivalent to ask that conditions~\eqref{kolmorovconditions} are satisfied. Trace preservation readily follows since for any $\rho=\sum_i r_i R_i$, $R_i$ being its eigenprojectors, $\Tr(\mathcal{L}[\rho])=\sum_{j} r_j \sum_i \Tr(R_i\mathcal{L}[R_j])=0$.  On the other hand, given any pair of orthogonal projectors $P$ and $Q$, $P\,Q=0$, let their spectral representation be $Q=\sum_{i\in I} q_i P_i, P=\sum_{i\in J} p_j P_i$, with $I \cap J=0$, so that
	$$
	\Tr(P\mathcal{L}[Q])=\sum_{i\in I, j\in J}  p_j\, q_i \Tr({P}_j\mathcal{L}[P_i])\ge0\,,	$$
	which completes the proof.
\end{proof}

\subsection{Classical from quantum P-divisible families}
\label{timedep:sec}

We now consider the case of time-dependent generators as in~\eqref{eq0} and positive intertwiners as in~\eqref{intertwiners}, namely the case of $P$-divisible families of $CPTP$ maps.
Theorem \ref{kossakowskiconditions} generalizes as follows~\cite{ColloquiumBreuer,ChrusReview22}.

\begin{theorem}\label{timekossakowskiconditions}
Let $\mathcal{L}_t$ be the generator of a $CPTP$ dynamical map $\{\Lambda_t\}_{t\geq 0}$. 
Then $\{\Lambda_t\}_{t\ge0}$ is $P$-divisible iff $\Tr(Q\mathcal{L}_t[P])\ge0$, 	for all $t\ge0$, and for any pair of mutually orthogonal projectors $P\perp Q$.
\end{theorem}

\begin{proof}
	The ``only if'' part follows straightforwardly from the positivity of $\Lambda_{t,s}$, since for all $s\ge0$
	\begin{equation}
		\Tr(Q\mathcal{L}_s[P])=\lim_{\epsilon\to 0}\frac{1}{\epsilon}{\Tr(Q\Lambda_{s+\epsilon,s}[P])}\ge0 .
	\end{equation}
	For the ``if'' part, rewrite $\mathcal{T}_{\leftarrow}e^{\int_{s}^t \dd{s}\mathcal{L}_s }$ by means of the time-splitting formula \cite{RivasHuelga}
	\begin{align}\label{timesplitting}
		\Lambda_{t,s} =\lim_{\max\limits_{k}
			 \delta t_k\to 0} \prod_{k=n-1}^0 e^{\delta t_k	\,\mathcal{L}_{t_k}}\,,
	\end{align}
	 with $t\equiv t_n \ge t_{n-1} \ge \dots\ge t_0\equiv s$ and $\delta t_k\equiv t_{k+1}-t_k$.
	 Fix $k\in\mathbb{N}$, $0\le k\le n-1$. Since $\Tr(Q\mathcal{L}_{t_k}[P])\ge0$ for all $P\perp Q$,   $t\mapsto e^{({t-t_k})\mathcal{L}_{t_k}}$ is a positive map for all $t\ge t_k$, by Theorem~\ref{kossakowskiconditions}. 
	 Positivity of $\Lambda_{t,s}$ follows from noting that the composition of positive maps is positive and that the set of positive maps is closed (for finite level systems, the closure is with respect to the norm topology on the space of linear maps of $\Md{d}$ onto itself) ~\cite{PaulsenPositive}.
\end{proof}
The classical analogue of the family of $P$-divisible $CPTP$ maps is a one-parameter family of stochastic matrices such that
\begin{equation}
	T(t)=T(t,s)T(s),
\end{equation}
for all $t\ge s\ge0$, where $T(t,t)=\mathds{1}$ and $T(t,s)$ is a stochastic matrix, $T_{ij}(t,s)\ge0$, $\sum_i T_{ij}(t,s)=1$. The classical time-local master equation has the form
\begin{equation}
	\dot T(t)=L(t) T(t).
\end{equation}

The quantum rendering of classical processes within a quantum setting outlined in Section~\ref{qcsem:sec} leads to the following generalizations of the Kolmogorov
conditions~\eqref{kolmorovconditions} and of Corollary~\ref{corollary1}.

\begin{proposition}
	\label{classicaldivis}
	A matrix $L(t)$ generates a classical $P$-divisible dynamics $T(t)$ iff
	\begin{equation}
		L_{ij}(t)\ge0 \quad \forall\,i \ne j, \qquad \sum_i L_{ij}(t)=0.
	\end{equation}
\end{proposition}
If $T(t)$ is an invertible matrix, writing $L(t)=\dot T(t) T(t)^{-1}$ and $T(t,s)=T(t)T(s)^{-1}$ provides a time dependent version of Corollary \ref{corollary1}.
\begin{corollary}
	Let $\mathcal{L}_t$ be the generator of a $CPTP$ dynamical map $\{\Lambda_t\}_{t\geq 0}$. 
Then $\{\Lambda_t\}_{t\ge0}$ is $P$-divisible iff the $n\times n$ matrix defined by $L_{ij}(t)=\Tr(P_i\mathcal{L}_t[P_j])$ generates a classical $P$-divisible dynamics for any orthonormal set of rank-1 projectors $\{P_i\}_i$.
\end{corollary}
Therefore, any quantum dynamics $\family{\Lambda_t}$ is $P$-divisible iff its generator defines classically $P$-divisible processes for any choice of orthonormal basis. \par

However, as discussed in Section~\ref{sec:Comparison}, one may derive a classical continuous stochastic process not from the quantum generator, but directly from the generated quantum maps, as in~\eqref{classicaldyn-t}.
When this reduction procedure is adopted, the relation between $P$-divisibility of $\Lambda_t$ and that of its classical reduction $T(t)$ is non-trivial and, in particular, classical $P$-divisibility may be lost. The following Section~\ref{sec:BFI_interpretation} is devoted to a general interpretation of such effect in terms of $BFI$. In the remainder of the paper, we will focus exclusively on classical processes of the type~\eqref{classicaldyn-t}, analysing in detail instances when $P$-divisibility is either preserved or lost through the reduction procedure.

\begin{remark}
	Notice that, from~\eqref{classicaldyn-t}, when the matrices $T(t)$ are invertible, the classical intertwiners are given by $T(t,s)=T(t)T(s)^{-1}$. In general, they have no simple connection to the quantum intertwiners $\Lambda_{t,s}=\Lambda_t\,\Lambda_s^{-1}$. A particular case is when the dynamics $\Lambda_t$ is of the kind
	$	\Lambda_{t,s} \,\mathbb{P}=\mathbb{P}\,\Lambda_{t,s}\, \mathbb{P},$ with $\mathbb{P}$ as in~\eqref{decoheringmap1},
so that
	one can identify a stochastic propagator as $T_{ik}(t,s)=\Tr(P_i\Lambda_{t,s}[P_k])$. Such maps are a subset of the so-called non-coherence-generating-and-detecting dynamics considered in~\cite{milz2020non} to characterize the classicality of a Markovian process.
\end{remark}

\section{Coherence-assisted backflow of information}
\label{sec:BFI_interpretation}

In the following, we consider a $P$-divisible quantum  evolution $\Lambda_t$ and its classical reduction obtained through~\eqref{classicaldyn-t} in order to connect the property of $P$-divisibility of the so-obtained classical process with the notion of information flows.

In the so-called $BLP$ approach to quantum non-Markovianity~\cite{BLP}, the Holevo-Helstrom distinguishability is interpreted as ``internal information'' of a quantum system as embodied by two states $\rho$ and $\sigma$ under the dynamics $\Lambda_t$:
\begin{equation}\label{internal_information}
	\mathcal{I}^q_t(\rho,\sigma;\mu)=\normt{\Lambda_t[\Delta_\mu(\rho,\sigma)]} \,,
\end{equation}
where
\begin{equation}
    \label{HelstromMAT}
    \Delta_\mu(\rho,\sigma)=\mu \rho-(1-\mu)\sigma
    \end{equation}
is the so-called Helstrom matrix, with $\mu\in [0,1]$ being the prior probability in preparing $\rho$ and $\sigma$, so that \eqref{internal_information} reduces for $\mu=1/2$ to the standard trace distance
\begin{equation}
\label{traceDIST}
D(\rho,\sigma)=\frac{1}{2}\normt{\rho-\sigma}\ .
\end{equation}
In the classical setting as in Section~\ref{qcsem:sec}, the Helstrom matrix becomes $\bs{\delta}_\mu(\bs{p},\bs{q})=\mu \bs{p}-(1-\mu)\bs{q}$, with $\bs{p},\bs{q}$ probability vectors, and its trace norm reduces to the so-called Kolmogorov distance
\begin{equation}
\norm{\bs{\delta}_\mu(\bs{p},\bs{q})}_{\ell_1}=\sum_i|\mu p_i-(1-\mu)\,q_i|\,,
\end{equation}
with respect to the $\ell_1$-norm $\norm{\bs{x}}_{\ell_1}=\sum_i \abs{x_i}$. Therefore, the internal information
$\mathcal{I}^{cl}_t(\boldsymbol{p},\boldsymbol{q};\mu)$ of the classical reduction is given by
\begin{equation}
 \mathcal{I}^{cl}_t(\boldsymbol{p},\boldsymbol{q};\mu)
=\norm{T(t) \bs{\delta}_\mu(\bs{p},\bs{q})}_{\ell_1}\ .
\end{equation}
For invertible quantum dynamics, $P$-divisibility is known to be equivalent \cite{CKR} to the monotonicity in time of $\mathcal{I}^q_t(\rho,\sigma;\mu)$,
\begin{equation}\label{BLP_condition}
	\partial_t \normt{\Lambda_t[\Delta_\mu(\rho, \sigma)]} \le 0\,,
\end{equation}
for any choice of the initial Helstrom matrix (not true in the non invertible case~\cite{Rivas_counterexample}). 
The typical interpretation of~\eqref{BLP_condition} is that information about the initial distinguishability of any pair of the system states is lost by the dissipative effect of the environment and never retrieved.
Similarly, for classical invertible dynamics, $P$-divisibility is equivalent to
\begin{equation}\label{classical_noBFI}
	\partial_t \norm{T(t)\boldsymbol{\delta}_\mu}_{\ell_1} \le 0\,,
\end{equation}
for all choices of $\boldsymbol{\delta}_\mu \in \mathbb{R}^d$.
We stress that lack of $P$-divisibility straightforwardly implies  a violation of \eqref{classical_noBFI} only if $T(t)$ is invertible.
Furthermore, in the usual open party scenario, a system is immersed in an environment and the compound state then evolves unitarily, $\rho_{SE}(t)=\mathcal{U}^{SE}_t[\rho\otimes\rho_E]$, so that the distinguishability of two system-environment states is constant and one defines accordingly an ``external information'' relative to two states $\rho_S$ and $\sigma_S$ of a quantum system $S$ as
\begin{align*}
 &\mathcal{E}^q_t(\rho_S,\sigma_S;\mu)\\
 &=\normt{\mathcal{U}_t^{SE}[\Delta_{\mu}(\rho_S,\sigma_S)\otimes\rho_E]}
 -\mathcal{I}^q_t(\rho_S,\sigma_S;\mu)\\
 &=\normt{\Delta_{\mu}(\rho_S,\sigma_S)\otimes\rho_E}
 -\mathcal{I}^q_t(\rho_S,\sigma_S;\mu)\ .
\end{align*}
Indeed, the trace-norm is invariant under unitary transformations, yielding
$\dot{\mathcal{E}^q_t}(\rho_S,\sigma_S;\mu)=-\dot{\mathcal{I}^q_t}(\rho_S,\sigma_S;\mu)$. Thus,
an increase of the internal information leads necessarily to a decrease of the external one. In the $BLP$ approach, such a recovery of distinguishability  is interpreted as backflow of information from the environment to the system. Such an interpretation, which involves the degrees of freedom of the environment, is also supported by the observation~\cite{ColloquiumBreuer,GTD_bound,HolevoSkew} that the difference of internal quantum informations at different times,
$$
\Delta \mathcal{I}^q_{t,s}(\rho_S,\sigma_S;\mu):=\mathcal{I}^q_t(\rho_S,\sigma_S;\mu)-\mathcal{I}^q_s(\rho_S,\sigma_S;\mu)\ ,
$$
can be upper-bounded by means of trace distances~\eqref{traceDIST} involving the marginal density matrices
$\rho_{S}(t)=\Tr_E\rho_{SE}(t)$, $\rho_{E}(t)=\Tr_S\rho_{SE}(t)$. Namely,
\begin{align}\label{boundnormal}
	\Delta \mathcal{I}^q_{t,s} \le& \,\, 2\mu\, D(\rho_{SE}(s),\rho_S(s)\otimes \rho_E(s))  \nonumber\\&+2(1-\mu) D(\sigma_{SE}(s),\sigma_S(s)\otimes \sigma_E(s))\,\\
	&+2\min\{\mu,1-\mu\}D(\rho_{E}(s), \sigma_E(s))\ .  \nonumber
\end{align}
It follows that $\Delta\mathcal{I}^q_{t,s}(\rho_S,\sigma_S;\mu)>0$ can only occur if system environment-correlations are present at time $s$ or if the environment marginals differ at time $s$.  \par
Let us now consider the decohering map $\mathbb{P}$ in~\eqref{decoheringmap1} and a Helstrom matrix $\Delta_\mu(\rho_{\boldsymbol{p}},\sigma_{\boldsymbol{q}})\in\mathcal{P}$ with
$\rho_{\boldsymbol{p}}=\sum_i p_i P_i$ and $\sigma_{\boldsymbol{q}}=\sum_i q_i P_i$.

Although the chosen density matrices commute and are incoherent classical states with respect to $\mathcal{P}$, the dynamics $\Lambda_t$ takes them out of the
commutative algebra $\mathcal{P}$ and generates non-vanishing coherences, while the classical reduction to $\mathcal{P}$ eliminates them. Notice that the classical internal information of the classical reduction to $\mathcal{P}$ of the quantum dynamics reads
\begin{equation}\label{classicalInternalinfo}
\mathcal{I}_{t}^{cl}(\boldsymbol{p},\boldsymbol{q};\mu)=\norm{T(t)\boldsymbol{\delta}_\mu(\boldsymbol{p},\boldsymbol{q})}_{\ell_1}=\normt{\mathbb{P}\Lambda_t[\Delta_\mu(\rho_{\boldsymbol{p}},\sigma_{\boldsymbol{q}})]}\ .
	\end{equation}
Let us then introduce the difference between the quantum and classical internal information relative to two classical states: \begin{eqnarray}
\nonumber
&&\hskip-.5cm
\mathcal{C}_t(\boldsymbol{p},\boldsymbol{q};\mu)= \mc{I}^q_t(\rho_{\boldsymbol{p}},\rho_{\boldsymbol{q}};\mu)-\mc{I}^{cl}_t(\boldsymbol{p},\boldsymbol{q};\mu)\\
\label{coherent_information}
&&		
=\normt{\Lambda_t[\Delta_\mu(\rho_{\boldsymbol{p}},\sigma_{\boldsymbol{q}})]}-\normt{\mathbb{P}\Lambda_t[\Delta_\mu(\rho_{\boldsymbol{p}},\sigma_{\boldsymbol{q}})]}\ .
\end{eqnarray}
Since $\mathbb{P}$ is a contraction, the quantity $\mathcal{C}_t(\boldsymbol{p},\boldsymbol{q};\mu)$ cannot be negative; moreover, it allows to write the quantum internal information of two quantumly evolving classical states as
\begin{equation}
\label{cohintinf}
\mathcal{I}^q_t(\rho_{\boldsymbol{p}},\rho_{\boldsymbol{q}};\mu)=\mathcal{I}_t^{cl}(\boldsymbol{p},\boldsymbol{q};\mu)+
\mathcal{C}_t(\boldsymbol{p},\boldsymbol{q};\mu)\ .
\end{equation}
It is thus appropriate to name $\mathcal{C}_t(\boldsymbol{p},\boldsymbol{q};\mu)$ \textit{coherent internal information}. Then, the $P$-divisibility of $\Lambda_t$ implies that $\mathcal{I}^q_t(\rho_{\boldsymbol{p}},\rho_{\boldsymbol{q}};\mu)$
monotonically decreases in time, namely the information contained in the system leaks towards the environment and never comes back. Hence, for $t \ge s \ge 0$,
\begin{equation}\label{balance_cl}
	\mathcal{I}_t^{cl}(\boldsymbol{p},\boldsymbol{q};\mu)+
\mathcal{C}_t(\boldsymbol{p},\boldsymbol{q};\mu) \le  \mathcal{I}_s^{cl}(\boldsymbol{p},\boldsymbol{q};\mu)+
\mathcal{C}_s(\boldsymbol{p},\boldsymbol{q};\mu)\ .
\end{equation}
Moreover, one can upper-bound the variation of the classical internal information between times $s$ and $t\ge s$ as (see Appendix~\ref{app2:bound})\begin{eqnarray}
\nonumber
&&
\Delta \mc{I}_{t,s}^{cl}(\boldsymbol{p},\boldsymbol{q};\mu):=\mathcal{I}_t^{cl}(\boldsymbol{p},\boldsymbol{q};\mu)-\mathcal{I}_s^{cl}(\boldsymbol{p},\boldsymbol{q};\mu)\\
 &&\hskip0.8cm \le \, \mathcal{C}_s(\bs{p},\bs{q};\mu) \nonumber\\
 && \hskip 0.8cm \le \, \mu \, C_{\ell_1}(\Lambda_s[\rho_{\bs{p}}]) \, + \,(1-\mu)\, C_{\ell_1}(\Lambda_s[\rho_{\bs{q}}]) \,,
\label{bound2_coh}
\end{eqnarray}
where
$C_{\ell_1}$ denotes the so-called $\ell_1$ norm of coherence of a state~\cite{QuantifyingCoherence,ColloquiumCoherence},
\begin{equation}
C_{\ell_1}(\rho):=\sum_{i\ne j}\abs{\rho_{ij}} \,,
\end{equation}
with respect to the classical subalgebra $\mathcal{P}$. Therefore, $C_{\ell_1}(\Lambda_s[\rho_{\bs{p}}])$ and $C_{\ell_1}(\Lambda_s[\rho_{\bs{q}}])$ measure the amount of coherence produced by the dynamics at time $s$ acting on the diagonal states $\rho_{\bs{p}}$ and $\rho_{\bs{q}}$.

In analogy with~\eqref{boundnormal}, from~\eqref{bound2_coh}, we can thus interpret the revival of the classical internal information, $\Delta \mc{I}_{t,s}^{cl}(\boldsymbol{p},\boldsymbol{q};\mu)> 0$, between times $s$ and $t\geq s$ as a classical backflow of information, which can occur only if a certain degree of quantum coherence has been built up to time $s$ in the quantumly evolving classical states. Hence, in the classical reduction of the quantum evolution of pairs of commuting quantum states, quantum coherences play an information storing role as the environment does in the quantum scenario, as emerges by comparing the r.h.s of~\eqref{bound2_coh} and \eqref{boundnormal}.

\section{Unital qubit dynamics}\label{sec:unitalqubit}

In this Section, we investigate the $P$-divisibility of classical reductions of the dynamical map obtained by~\eqref{classicaldyn-t}. In particular, we will argue that the question of whether $P$-divisibility of the quantum process is inherited by its reduction  becomes particularly relevant when one considers purely dissipative dynamics with generators not containing  commutators with Hamiltonians (see~\eqref{GKSL}), which would provide non-divisible  classical processes.
Moreover, we shall mostly focus on qubit unital dynamics for which the $P$-divisibility of both $\Lambda_t$ and its classical reduction  $T(t)$ can be controlled.

Let $\Lambda_t$ be a unital qubit dynamics,  $\Lambda_t[\mathds{1}]=\mathds{1}$, and $P_0$, $P_1$ two orthogonal projectors generating a maximally Abelian sub-algebra $\mathcal{P}\subset M_2(\mathbb{C})$. As in Example \ref{example:classstochdoesntimplypos}, the structure of the classical process $T(t)$
obtained from $\Lambda_t$ by means of~\eqref{classicaldyn-t} is determined only by $T_{00}(t)$:
\begin{equation}\label{tdepstochasticmatrix}
	T(t)=\begin{pmatrix}
		T_{00}(t) & 1-T_{00}(t)\\
		1-T_{00}(t) & T_{00}(t)
	\end{pmatrix} ,
\end{equation}
where $T_{00}(t)=\Tr(P_0\Lambda_t[P_0])$. If $T_{00}(t)\ne\frac{1}{2}$, $T(t)$ is invertible and one computes the classical generator as
\begin{equation}
	\label{matgen}
	L(t)=\dot T(t) T(t)^{-1}=\frac{\dot T_{00}(t)}{2T_{00}(t)-1}\begin{pmatrix}
		1 & -1\\
		-1 & 1
	\end{pmatrix},
\end{equation}
Thus, from Proposition \ref{classicaldivis}, $L(t)$ generates of a divisible classical process iff
\begin{equation}\label{condition_classicalprocess0}
	f_t:=\frac{\dot T_{00}(t)}{2T_{00}(t)-1} \le 0\qquad \forall t\geq 0\ .
\end{equation}
As in Example \ref{example:classstochdoesntimplypos}, we thus introduce the $3\times 3$ matrix representation of $\Lambda_t$
\begin{equation}\label{Pauli_repDyn}
	\widetilde{\Lambda}_{ij}(t)=\frac{1}{2}\Tr(\sigma_i\Lambda_t[\sigma_j]),
\end{equation}
and analogously $\wt{\mathcal{L}}(t)$ for $\mathcal{L}_t$, satisfying the master equation
$	\dot{\widetilde{\Lambda}}(t)=\widetilde{\mathcal{L}}(t) \widetilde{\Lambda}(t)$,
whose formal solution is given by  a time-ordered exponentiation (see~\eqref{timeordering}):
\begin{equation}\label{timeordering_tilde}
	\widetilde{\Lambda}(t)=\mathcal{T}_\leftarrow e^{\int_0^t\dd{s}\widetilde{\mathcal{L}}(s)}.
\end{equation}
Using the Bloch representation of Example~\ref{example:classstochdoesntimplypos}, one rewrites
\begin{equation*} \label{rewritingmatricial}
	T_{00}(t)=\frac{1}{2}+\frac{1}{2}\mel{\boldsymbol{n}}{\widetilde{\Lambda}(t)}{\boldsymbol{n}}, \ \dot T_{00}(t)=\frac{1}{2}\mel{\boldsymbol{n}}{\widetilde{\mathcal{L}}(t)\widetilde{\Lambda}(t)}{\boldsymbol{n}}\ .
\end{equation*}
Then,~\eqref{condition_classicalprocess0} reads
\begin{equation}
\label{condition_classicalprocess}
	f_t=\frac{1}{2}\frac{\mel{\boldsymbol{n}}{\widetilde{\mathcal{L}}(t)\widetilde{\Lambda}(t)}{\boldsymbol{n}}}{\mel{\boldsymbol{n}}{\widetilde{\Lambda}(t)}{\boldsymbol{n}}} \le 0\,, \qquad\forall t\geq 0\ .
\end{equation}
Similarly, $\Lambda_t$ is $P$-divisible iff
\begin{equation}\label{Pdtilde}
	\Tr(P_0\mathcal{L}_t[P_1])=	-\frac{1}{2}\mel{\boldsymbol{n}}{\widetilde{\mathcal{L}}(t)}{\boldsymbol{n}}\ge0\qquad \forall \boldsymbol{n}\in\mathbb{R}^3\ ,
\end{equation}
namely iff the symmetric part of the generator in the Bloch representation has negative eigenvalues
\begin{equation}\label{P-divunital}
	-(\widetilde{\mathcal{L}}(t)+\widetilde{\mathcal{L}}^T(t)) \ge 0\,.
\end{equation}

\par
In the following, we consider the behaviour of $f_t$ under two oppositely behaving dynamics.
The first one is a purely unitary qubit rotation which  cannot give rise to a classically divisible process, while the second is that of a purely dissipative Pauli dynamics for which $P$-divisibility of the dynamics is equivalent to that of its classical reduction.

\begin{example}\label{ex:consthamiltonian}

Consider the qubit Hamiltonian $H=1/2\boldsymbol{\omega}\cdot\boldsymbol{\sigma}$, $\boldsymbol{\omega}\in\mathbb{R}^3$, and the evolution $\mathcal{U}_t[\rho]=U_t \,\rho \, \daga{U}_t$, $U_t=e^{-itH}$.
The classical reduction of its generator $\mathcal{L}[\cdot]=-i[H\,,\cdot\,]=0$ to any commutative sub-algebra $\mathcal{P}\subset M_2(\mathbb{C})$ vanishes. Indeed, $\Tr(PHQ)=\Tr(PQH)=0$. However, the Bloch representation of  $\mathcal{U}_t$ acts as $3\times3$ rotation matrix
$\wt{\mathcal{U}}(t)=e^{t\wt{\mathcal{L}}}$ with generator
\begin{equation}
	\wt{\mathcal{L}}=\omega\,\begin{pmatrix}
		0 & -r_3 & r_2 \\
		r_3 & 0  & -r_1 \\
		-r_2 & r_1  & 0 \\
	\end{pmatrix}\ ,
\end{equation}
where $\omega=\|\boldsymbol{\omega}\|$ and $\boldsymbol{r}=(\omega_1,\omega_2,\omega_3)/\omega\in\mathbb{R}^3$ is a unit vector.
Since $\wt{\mathcal{L}}^3=-\omega^2 \wt{\mathcal{L}}=\omega^2\wt{\mathcal{L}}^T$, one rewrites
\begin{equation}\label{SO3_Rodrigues}
\wt{\mathcal{U}}(t)=e^{t\wt{\mathcal{L}}}=\mathds{1}+\frac{\sin(\omega t)}{\omega}\,\wt{\mathcal{L}}+\frac{1-\cos(\omega t)}{\omega^2}\, \wt{\mathcal{L}}^2\ .
\end{equation}

Also, $\mel{\boldsymbol{n}}{\wt{\mathcal{L}}}{\boldsymbol{n}}=0$,
$\mel{\boldsymbol{n}}{\wt{\mathcal{L}}^{2}}{\boldsymbol{n}}=-\|\wt{\mathcal{L}}\ket{\boldsymbol{n}}\|^2$.
Given a projector with Bloch vector $\boldsymbol{n}$, with $\theta$ the angle between the latter and $\boldsymbol{r}$, one has $\|\wt{\mathcal{L}}\ket{\boldsymbol{n}}\|=\omega \sin(\theta)$ and
\begin{eqnarray*}
T_{00}(t)&=&\Tr\Big(P_0\,\mathcal{U}_t[P_0]\Big)=\frac{1}{2}\Big(1+\mel{\boldsymbol{n}}{\wt{\mathcal{U}_t}}{\boldsymbol{n}}\Big)
\\
&=&\frac{1}{2}(1+\cos^2(\theta)+\cos(\omega t)\sin^2(\theta))\ ,
\end{eqnarray*}
so that \eqref{condition_classicalprocess} is rewritten as follows,
\begin{equation}\label{unitary_f}
	f_t=-\frac{\omega}{2}\frac{\sin(\omega t)\sin^2(\theta)}{\cos^2(\theta)+\cos(\omega t)\sin^2(\theta)}\ .
\end{equation}
Thus, the sign of $f_t$ changes for all $\theta\in(0,\pi)$, unless $\theta=0,\pi$ (when $\ket{\boldsymbol{n}}$ is an eigenstate of $H$) yielding $f_t=0$.
Notice that the denominator in $f_t$ is the determinant of the classical stochastic matrix $T(t)$ which is thus
invertible if
 $\cos(2\theta)>0$, namely, if $\theta \in (0,\pi/4)\cup(3\pi/4,\pi)$. In this case, classical $P$-divisibility breaks when $\dot T_{00}(t)$ becomes positive. 
 As described in Section \ref{sec:BFI_interpretation}, such loss of $P$-divisibility for an invertible classical reduction $T(t)$ can be interpreted in terms of $BFI$.
\end{example}

\begin{example}
\label{ex:Pauli_dyn}
Let $\Lambda_t$ be the Pauli dynamics generated by
\begin{equation}\label{Pauligenerator}
	\mathcal{L}_t[\rho]=\frac{1}{2}\sum_{k=1}^3 \gamma_k(t)(\sigma_k\rho\sigma_k-\rho).
\end{equation}
Notice that  ${[\mathcal{L}_t,\mathcal{L}_s]=0}$, yielding the exponential solution  $\Lambda_t=e^{\int_0^t\dd{s}\mathcal{L}_s}$. Moreover, the representation in the Pauli representation of the generator is diagonal,
\begin{equation}
\widetilde{\mathcal{L}}_{ij}(t)[\rho]=-\Gamma_i(t) \,\delta_{ij}\,, \qquad \Gamma_i(t)=\sum_{k\ne i} \gamma_k(t)\,,
\end{equation}
and necessary and sufficient for $P$-divisibility follow from~\eqref{P-divunital},
\begin{equation}\label{pauliPdivisibility}
 \Gamma_k (t)\ge0 \,, \quad k=1,2,3\,.	
\end{equation}
$\widetilde{\Lambda}(t)$ is also diagonal, with strictly positive eigenvalues
$$
\widetilde{\Lambda}_{ij}(t)=\lambda_i(t) \delta_{ij} \,,\qquad \lambda_i(t)=e^{-\int_{0}^t \dd{s} \Gamma_i(t)}\ .
$$
so that
\begin{align}
\label{fPt}
f_t=\frac{1}{2}	\frac{\mel{\boldsymbol{n}}{\widetilde{\mathcal{L}}(t)\widetilde{\Lambda}(t)}{\boldsymbol{n}}}{\mel{\boldsymbol{n}}{\widetilde{\Lambda}(t)}{\boldsymbol{n}}}=-\frac{1}{2}\frac{\sum_i \Gamma_i(t) \lambda_i(t) n_i^2 }{\sum_j \lambda_j(t) n_j^2} \le 0
\end{align}
for all $\boldsymbol{n}$.
Thus, $T(t)$  is $P$-divisible for all choices of the reference basis iff $\Lambda_t$ is $P$-divisible.
\end{example}

Notice that, for $\gamma_i(t)=const$, the Pauli generator includes the case of a unital purely dissipative GKSL evolution (the time-independent version of \eqref{purely_dissipativeunital}), which thus yields $P$-divisible classically reduced dynamics.

\begin{example}
	\label{rem3}
	Adding a Hamiltonian to a Pauli generator $\mathcal{L}_t^P$ as in~\eqref{Pauligenerator} may generally lead to a positive $f_t$, so that the process $T(t)$ is not divisible. Indeed, consider
	\begin{equation}
		\mathcal{L}_t[\rho]=-i[\sigma_z,\rho]+\mathcal{L}_t^P[\rho],
	\end{equation}
	and take, for the sake of simplicity, $\Gamma_1(t)=\Gamma_2(t)\equiv\Gamma(t)$, so that the Hamiltonian generator commutes with the Pauli one. Hence, the evolution will have the form
	\begin{equation}
    \Lambda_t=\mathcal{U}_t\,\Lambda_t^P \implies 	\widetilde{\Lambda}(t)=
    \wt{\mathcal{U}}(t)\,\widetilde{\Lambda}^P(t)\,.
	\end{equation}
		Using~\eqref{condition_classicalprocess} and~\eqref{SO3_Rodrigues}, it follows that
	$f_t=f_t^P+f_t^H$,
	where $f_t^P$ is as in~\eqref{fPt} and always negative, while $f_t^H$ is as in~\eqref{unitary_f} and generally oscillates between positive and negative values. In particular, $P$-divisibility is lost due to the divergence of the $f_t^H$ for some $\bs{n}$. On the other hand, for those $\bs{n}$ that lead to an invertible $T(t)$, $f_t^H$ is bounded, so that $f_t$ can stay negative for all times  for large enough Pauli rates $\Gamma_i(t)$.
\end{example}

The previous examples suggest that the presence of the commutator in the generator is in general responsible for the loss of $P$-divisibility for the classically reduced process defined by~\eqref{classicaldyn-t}, while the classical reduction of
Pauli dynamics preserves $P$-divisibility.
Therefore, in order to study which $P$-divisible quantum dynamics $\Lambda_t$ give certainly rise to classically divisible processes defined by~\eqref{classicaldyn-t}, one better focus upon purely dissipative generators.

From~\eqref{GKSL}, it follows that the time-local generator of a trace and Hermiticity preserving, qubit dynamics can be always expressed in the form~\cite{GoriniKossSud}
\begin{equation}\label{timedep_qubit}
	\mathcal{L}_t[\rho]=-i[H(t),\rho]+\mathcal{D}_t[\rho],
\end{equation}
with dissipative part
\begin{equation}\label{timedependent_qubitgenerator}
	\mathcal{D}_t[\rho]=\frac{1}{2}\sum_{i,j=1}^3 K_{ij}(t)\left(\sigma_i\rho\sigma_j-\frac{1}{2}\acomm{\sigma_j\,\sigma_i}{\rho}\right) \ .
\end{equation}

For a unital qubit dynamics
the following then holds.

\begin{lemma}\label{lemma1}
	A qubit dynamics generated by $\mathcal{L}_t$ is unital iff the Kossakowski matrix $K(t)$ is real and symmetric, so that $\mathcal{D}_t$ can be written in the diagonal form
	\begin{equation}\label{purely_dissipativeunital}
		\mathcal{D}_t[\rho]=\frac{1}{2}\sum_i \gamma_i(t)\left(\sigma_i(t)\rho\sigma_i(t)-\rho\right),
	\end{equation}
	with respect to a set of matrices (generally time dependent) $\{\sigma_i(t)\}_{i=1}^3$ that fulfil the Pauli algebra. Moreover, the unital dynamics is purely dissipative, namely $\mathcal{L}_t=\mathcal{D}_t$, iff the time-independent Bloch-representation of the generator
	\begin{equation}
		\label{PaulirepL}
		\widetilde{\mathcal{L}}_{ij}(t)=\frac{1}{2}\Tr(\sigma_i\mathcal{L}_t[\sigma_j]),\quad i,j=1,2,3\ ,
	\end{equation}
	consists of a symmetric matrix.
\end{lemma}

\begin{proof}
	Let us recall that $\Lambda_t$ is unital if and only if $\mc{L}_t[\mathds{1}]=0$ for all $t \ge 0$. Indeed, if $\Lambda_t$ is unital, $0=\dot\Lambda_t[\mathds{1}]=\mc{L}_t\Lambda_t[\mathds{1}]=\mc{L}_t[\mathds{1}]$; conversely, if the generator kills the identity at all times, unitality follows from the Dyson series~\eqref{timeordering}. Hence, in the qubit case, plugging the identity into~\eqref{timedependent_qubitgenerator}
$\Lambda_t$ is unital iff
	\begin{equation}
		\sum_{i,j=1}^3K_{ij}(t)[\sigma_i,\sigma_j]=0,
	\end{equation}
	which is satisfied iff $K_{ij}(t)=K_{ji}(t)={K_{ij}^*}(t)$. Then, the real symmetric Kossakowski matrix $K(t)$ can be diagonalized by means of an orthogonal transformation,
	$$
	K(t)=O(t)\gamma(t)O(t)^T\ ,\quad  \gamma(t)=\textrm{diag}\{\gamma_1(t),\gamma_2(t),\gamma_3(t)\}\ .
	$$
	Therefore, $O(t)O^T(t)=O^T(t)O(t)=\mathds{1}$ guarantees that the matrices $\sigma_i(t)=\sum_j O_{ji}(t)\sigma_j$ satisfy the Pauli algebra.
	Moreover, the Bloch representation~\eqref{PaulirepL} of the dissipative part $\mathcal{D}_t$ of the generator reads
	\begin{align}
		\label{PaulirepD}
		\widetilde{\mathcal{D}}_{ij}(t)
		&=\frac{1}{2}\sum_k\gamma_k(t)\left[ \Tr\left(\sigma_i\sigma_k(t)\sigma_j\sigma_k(t)\right)-2\delta_{ij} \right]\nonumber\\
		&=\widetilde{\mathcal{D}}_{ji}(t)\ ,
	\end{align}
	where the last equality follows from the cyclicity property of the trace. Thus, an anti-symmetric contribution to the matrix representation of the full generator $\mathcal{L}(t)$ can only come from the commutator with the Hamiltonian. It follows that the unital dynamics is purely dissipative,  namely that $\widetilde{\mathcal{L}}(t)=\widetilde{\mathcal{D}}(t)$,  iff ${\widetilde{\mathcal{L}}(t)=\widetilde{\mathcal{L}}^T(t)}$.
\end{proof}

\begin{remark}
\label{remselfdual}
The dual, $\Lambda^\ddagger$, of a state-transformation $\Lambda$ is defined as
$$
\Tr(X\Lambda[\rho])=\Tr(\Lambda^\ddagger[X]\rho),
 $$
for all system states $\rho$ and system operators $X$. Since $\Lambda$ is trace-preserving $\Lambda^\ddagger$ is unital
$\Lambda^\ddagger[\mathds{1}]=\mathds{1}$. A state-transformation $\Lambda$ is self-dual iff $\Lambda=\Lambda^\ddagger$ so that
it is necessarily unital, while the reverse is obviously not true.
When $\Lambda_t$ is a qubit dynamical map, going to the Bloch representation, one then has that self-duality is equivalent to
  $\wt{\Lambda}(t)=\wt{\Lambda}^T(t)$. On the other hand, if the generator $\mathcal{L}_t$ of the qubit dynamics $\Lambda_t$
  is self-dual, $\mc{L}_t^{\phantom{\ddagger}}=\mc{L}_t^\ddagger$, it must be purely dissipative, $\mathcal{L}_t=\mathcal{D}_t$, and $\Lambda_t$ unital.
Indeed,  $\mc{L}_t[\mathds{1}]=\mc{L}_t^\ddagger[\mathds{1}]=0$ yields a symmetric Kossakowski matrix which in turn implies
the absence of the commutator with a Hamiltonian. The unitality of $\Lambda_t$ then follows from its Dyson expansion.
Moreover, in the qubit case, Lemma \ref{lemma1} implies that a purely dissipative generator of a unital dynamics is necessarily self-dual. In the following Section \ref{sec:class}, we shall provide a concrete qubit construction for a non self-dual unital dynamics arising from a self-dual generator. In particular, this cannot occur if the solution is of the form $\Lambda_t=e^{\int_0^t\dd{s}\mc{L}_s}$.
\end{remark}

It follows that if $\Lambda_t$ is unital and purely dissipative, $\widetilde{\mathcal{L}}(t)$ is real symmetric and then, from \eqref{P-divunital}, the $P$-divisibility of $\Lambda_t$ becomes equivalent to
$-\wt{\mathcal{L}}(t)\geq 0$. The result for Pauli dynamics can be then generalized to self-dual qubit dynamics with self-dual generators, for which $P$-divisibility is preserved by their classical reduction.
\begin{proposition}\label{prop_selfduality}
	Let $\Lambda_t^{\phantom{\ddagger}}=\Lambda_t^{\ddagger}$ be a self-dual, purely dissipative, invertible  qubit dynamics.
	Then, the associated classical stochastic process $T(t)$ is divisible if and only if $\Lambda_t$ is $P$-divisible.
\end{proposition}

\begin{proof}
	One has to check when~\eqref{condition_classicalprocess} holds. Since
	$\Lambda_t^{\phantom{\ddagger}}=\Lambda_t^\ddagger$, this  implies that their Bloch representations satisfy
	\begin{equation}\label{symmetricBloch}
		\wt{\Lambda}(t)=\wt{\Lambda}^T(t)\,.
		\end{equation}
Also, the assumed pure dissipativeness of $\Lambda_t$ entails that the generator itself is self-dual, $\mc{L}_t=\mc{L}_t^\ddagger$, so that $\wt{\mathcal{L}}^T(t)=\wt{\mathcal{L}}(t)$ (see Lemma~\ref{lemma1}).
Then, taking the time derivative of both members of~\eqref{symmetricBloch},
one gets
	\begin{equation}\label{commutLandSOl}
		\dot{\widetilde{\Lambda}}(t)=\widetilde{\mathcal{L}}(t)\widetilde{\Lambda}(t)=\widetilde{\Lambda}(t)\widetilde{\mathcal{L}}(t)\ ,
	\end{equation}
	or, equivalently, $[{\widetilde{\Lambda}(t)},{\dot{\wt{\Lambda}}(t)}]=0$. The invertibility of $\Lambda_t$ means that none of the real eigenvalues of $\wt{\Lambda}(t)=\wt{\Lambda}^T(t)$ can change sign with varying $t$. Since, at $t=0$, $\wt{\Lambda}_t=\mathds{1}$ all of them must remain positive. Therefore, one can consistently express the generator as a logarithmic derivative:
	\begin{equation}
		\wt{\mathcal{L}}(t)= \dot{\widetilde{\Lambda}}(t)\widetilde{\Lambda}(t)^{-1}=\widetilde{\Lambda}(t)^{-1}\dot{\widetilde{\Lambda}}(t)=\dv{t}\log\wt{\Lambda}(t)\ .
	\end{equation}
	Thus, $\wt{\Lambda}(t)=e^{\int_{0}^{t}\dd{s}\wt{\mathcal{L}}(s)}$ without time-ordering; furthermore,
	\begin{equation}\label{f_tproof}
		f_t=-\frac{1}{2}\frac{\mel{\boldsymbol{n}}{\sqrt{\widetilde{\Lambda}(t)}(-\widetilde{\mathcal{L}}(t))\sqrt{\widetilde{\Lambda}(t)}}{\boldsymbol{n}}}{\mel{\boldsymbol{n}}{\widetilde{\Lambda}(t)}{\boldsymbol{n}}} \le 0\,.
	\end{equation}
 The proof is completed by observing that, because of~\eqref{P-divunital}, the dynamics $\Lambda_t$ is $P$-divisible iff $-\wt{\mathcal{L}}(t)\geq 0$.
\end{proof}

\begin{remark}
\label{remNOTO}
The assumptions in Proposition~\ref{prop_selfduality} lead to a proper exponential solution $\wt{\Lambda}(t)=e^{\int_{0}^{t}\dd{s}\wt{\mathcal{L}}(s)}$ without asking for commuting generators at different times as in Example~\ref{ex:Pauli_dyn}. Indeed, a sufficient condition for the exponential to solve
 $\dot{\wt{\Lambda}}(t)=\wt{\mathcal{L}}(t)\wt{\Lambda}(t)$ is that $\Big[\wt{\mathcal{L}}(t)\,,\,\int_0^t{\rm d}s\,\wt{\mathcal{L}}(s)\Big]=0$ for all $t\geq 0$ (notice that, from the power series of the exponential,  \eqref{commutLandSOl} also follows).
 As an example, consider a generator of a unital dynamics of the form
 $$
 \wt{\mathcal{L}}(t)=\left\{
 \begin{matrix}
 \begin{pmatrix}
 -1&\cos t&0\cr
 \cos t&-1&0\cr
 0&0&-2\cr
     \end{pmatrix}& \quad 0\leq t\leq \pi\cr
-\begin{pmatrix}
    1&0&0\cr
    0&2&0\cr
    0&0&1
\end{pmatrix}& \quad \pi\leq t
     \end{matrix}
 \right.\ .
 $$
 $\wt{\mathcal{L}}(t)$ and $\wt{\mathcal{L}}(s)$ do not commute when $0\leq s\leq \pi$ and $t>\pi$, while they do when either $0\leq s,t\leq \pi$ or $s,t> \pi$; on the other hand, since
 $$
 \int_0^t{\rm d}s\,\wt{\mathcal{L}}(s)=\begin{pmatrix}
     -t&\sin t&0\cr\sin t&-t&0\cr0&0&-2t
     \end{pmatrix}
$$
when $0\leq t\leq \pi$ and, when $t\ge\pi$,
$$
\int_0^t{\rm d}s\,\wt{\mathcal{L}}(s)=
-\pi\begin{pmatrix}
     1&0&0\cr0&1&0\cr0&0&2
     \end{pmatrix}-(t-\pi)\begin{pmatrix}
     1&0&0\cr0&2&0\cr0&0&1
     \end{pmatrix}
     $$
    the condition $\Big[\wt{\mathcal{L}}(t),\int_0^t{\rm d}s\,\wt{\mathcal{L}}(s)\Big]=0$, for all $t\geq 0$, that avoids time-ordering is fulfilled. Moreover, $\wt{\mc{L}}(t)$ is symmetric (self-dual generator), so that the dynamics is purely dissipative. The exponential solution then entails the self-duality of the map. $-\wt{\mc{L}}(t)\ge0$ finally assures  $P$-divisibility of the generated dynamics, so that all the hypothesis of Proposition~\ref{prop_selfduality} are matched.
\end{remark}

\section{A class of orthogonally covariant dynamics}\label{sec:class}

In the following, we study whether the $P$-divisibility of a quantum qubit dynamics $\Lambda_t$ is inherited by its classical reductions $T(t)$ for non self-dual, purely dissipative dynamics. 
We focus upon the following family of maps,
	\begin{align}\label{class_notime}
		\Phi^{({A,\lambda,\mu})}[\rho]=&\sum_{i,j=0}^1 A_{ij} E_{ij} \rho E_{ji} + \lambda E_{00} \rho E_{11}+\overline{\lambda}\, E_{11} \rho E_{00} \nonumber\\  &+\mu \, E_{11}\rho^T  E_{00}+\overline{\mu}\, E_{00} \rho^T E_{11}  \,,
	\end{align}
where $E_{ij}=\ketbra{i}{j}$ are the matrix units associated with the basis of eigenvectors of $\sigma_3$, $\sigma_3\ket{0}=\ket{0}, \sigma_3\ket{1}=-\ket{1}$. The maps are not in general of Pauli type and
depend parametrically on  a triplet $(A,\lambda, \mu)$, given by a $2\times2$ matrix ${A=[A_{ij}]}$ and coefficients $\lambda, \mu \in \mathbb{C}$.

They were studied in the $d$-dimensional case in~\cite{SinghNecita2021}. In particular, they satisfy the following  group composition law,
 \begin{equation}\label{groupclass}
 	\Phi^{({A,\lambda,\mu})}\Phi^{({A',\lambda',\mu'})}=\Phi^{({AA',\,\lambda\lambda'+\mu\overline{\mu'},\,\lambda\mu+\mu\overline{\lambda'}})}\ .
 \end{equation}
Moreover, if $\mu=0$, the maps $\Phi^{(A,\lambda,0)}$ satisfy the diagonal unitary-covariance property
$$
\Phi^{(A,\lambda,0)}[U X U^\dagger]=U \Phi^{(A,\lambda,0)}[X] U^\dagger\ ,\quad
U=\sum_i e^{i \theta_i} E_{ii}\ .
$$
Conversely, if $\lambda=0$, the maps $\Phi^{(A,0,\mu)}$ are conjugate diagonal unitary-covariant, namely
$$
\Phi^{(A,0,\mu)}[U X U^\dagger]=\overline{U} \Phi^{(A,0,\mu)}[X] U^T\ .
$$
Lastly, if  both $\lambda,\mu$ are different from zero, the only symmetry left is with respect to rotations around the $z$ axis, corresponding to the diagonal orthogonal covariance
$$
\Phi^{(A,\lambda,\mu)}[O X O^T]=O \Phi^{(A,\lambda,\mu)}[X]O^T\ ,
$$
where $O=\sum_{i} o_i E_{ii}$, $o_i\in\{-1,1\}$.

Consider a time-dependent family of hermiticity and trace-preserving  maps within the above class,
$\Lambda_t=\Phi^{(A(t),\lambda_t,\mu_t)}$. Then, $A(t)$ is to be taken real and of the form
$A_{ij}(t)\in \mathbb{R}$ and $\sum_{i} A_{ij}(t)=1$, so that
\begin{equation}
A(t)=\begin{pmatrix}
	a_t & 1-b_t \\
	1-a_t & b_t
\end{pmatrix}\ ,\quad a_t, b_t\in\mathbb{R}\ .
\end{equation}
Asking that $\Lambda_{t=0}=\mathrm{id}_2$ yields $\lambda_0=1, \mu_0=0$ and $a_0=1$; whereas positivity can be proved to be equivalent to~\cite{SinghNecita2021}
\begin{align}
	&A_{ij}(t)\ge 0 \,,\nonumber\\ &\abs{\lambda_t}+\abs{\mu_t} \le \sqrt{a_t b_t}+\sqrt{(1-a_t)(1-b_t)}\,.
\end{align}
so that $A(t)$ has to be a stochastic matrix. On the other hand, the conditions for complete positivity can be obtained by imposing positivity of the associated Choi matrix,
\begin{equation}\label{Choi}
	\Lambda_t\otimes \mathrm{id}_2\bigg[\sum_{ij}E_{ij}\otimes E_{ij}\bigg] \ge 0\,,
\end{equation}
leading to the stronger necessary and sufficient conditions
\begin{equation}\label{CP_class}
	\abs{\lambda_t}\le \sqrt{a_t b_t} \,, \quad \abs{\mu_t} \le\sqrt{(1-a_t)(1-b_t)}\,.
\end{equation}
Recently, Schwartz-positivity for maps within this class has also been fully characterized~\cite{Bihalan}. Sufficient conditions for complete positivity in terms of the generator for this class of maps were also investigated in~\cite{HallBlocks}.
We shall now characterize the divisibility properties of the dynamics within class \eqref{class_notime}. Due to the composition property~\eqref{groupclass}, the generator ${\mathcal{L}_t=\dot{\Lambda}^{\phantom{.}}_t\Lambda_t^{-1}}$ will itself belong to the class~\eqref{class_notime}, ${\mathcal{L}_t=\Phi^{(B(t),\ell_t,m_t)}}$, with $B(t)$ of the form
\begin{equation}
	 B(t)=\dot A(t)A(t)^{-1}=\begin{pmatrix}
	 	-\gamma_-(t) & \gamma_+(t) \\
	 	\gamma_-(t) & -\gamma_+(t)
	 \end{pmatrix}
\end{equation}
where $\gamma_-(t)$, $\gamma_+(t)$ are related to $a_t$, $b_t$ via
\begin{equation*}
	\gamma_-(t)=-\frac{\dot a_t (1-b_t)+ \dot b_t a_t}{a_t+b_t-1} , \ \	\gamma_+(t)=-\frac{\dot a_t b_t+ \dot b_t (1-a_t)}{a_t+b_t-1},
\end{equation*}
and $l_t$, $m_t$ are related to $\lambda_t$, $\mu_t$ via
\begin{align}\label{ell_emm}
	l_t&=\frac{\dot\lambda_t\overline{\lambda_t}-\dot\mu_t\overline{\mu_t}}{\abs{\lambda_t}^2-\abs{\mu_t}^2} \,,  \qquad	m_t=\frac{\dot\mu_t\lambda_t-\dot\lambda_t{\mu}_t}{\abs{\lambda_t}^2-\abs{\mu_t}^2}  \,.
\end{align}
It is convenient to introduce time-dependent ``transversal'' and ``longitudinal'' rates,
\begin{equation}
	\Gamma_T(t):=-\Re(l_t)\,, \qquad \Gamma_L(t):={\gamma_+(t)+\gamma_-(t)}\,.
\end{equation}
and let also  $\omega(t):=-\Im(l_t)$.
Notice that $\Gamma_T(t)$ only depends on the absolute values of $\lambda_t$ and $\mu_t$,
\begin{equation}\label{Rel_explicit}
	\Gamma_T(t)=-\frac{1}{2}\frac{\partial_t \abs{\lambda_t}^2-\partial_t\abs{\mu_t}^2}{|\lambda_t|^2-|\mu_t|^2} \,.
\end{equation}
Necessary and sufficient conditions for $P$ and $CP$-divisibility of the dynamics can be then determined from the generator.

\begin{proposition}
\label{propP-div}
Let  ${\mathcal{L}_t=\Phi^{(B(t),l_t, m_t)}}$ be the generator of an evolution $\Lambda_t=\Phi^{(A(t),\lambda_t, \mu_t)}$. Then,
\begin{itemize}
	\item $\Lambda_t$ is $P$-divisible iff
	\begin{align}\label{nonunitalPD}
		&\gamma_{\pm}(t) \ge 0 \,,\\
		&\Gamma_T(t)-\frac{\Gamma_L(t)}{2}+\sqrt{\gamma_+(t)\gamma_-(t)} \ge \abs{m_t}\,;\label{nontrivial_PD}
	\end{align}
	\item $\Lambda_t$ is $CP$-divisible iff
	\begin{align}
		& \gamma_{\pm}(t) \ge 0 \,,
		\,\qquad {\gamma_+(t)\gamma_-(t)}\ge \abs{m_t}^2\,, \label{CPd1}\\
		& \Gamma_T(t) \ge \frac{\Gamma_L(t)}{2}\ .
  \label{CPd2}
	\end{align}
\end{itemize}
\end{proposition}

Though in a different form, this result has already been obtained in~\cite{CabreraDavalosGorin2019}; in Appendix~\ref{App1} we report the proof since it is somewhat simpler that in the quoted reference. Notice that condition~\eqref{CPd2} is the time-dependent version of the celebrated constraint between relaxation rates already discussed in seminal paper~\cite{GoriniKossSud}, which has been proved to be universal for qubit completely positive semigroups~\cite{KimuraQubit}, and recently conjectured to be universal for all completely positive dynamical semigroups~\cite{chruscinski2021universal}.

	\begin{remark}\label{techremark}
			The Bloch representation of $\Lambda_t$ as in \eqref{Pauli_repDyn} leads to the following dynamics of the
   Bloch vector:
			\begin{equation}
				\bs{n}_t = \widetilde{\Lambda}(t) \bs{n} + \bs{u}_t\,, \qquad \bs{u}_t=(0,0,a_t-b_t)\,,
			\end{equation}
			with
		\begin{equation*}\label{class_matrixrepr}
			\widetilde{\Lambda}(t)=\begin{pmatrix}
				\Re(\lambda_t+\mu_t)   &     \Im(\lambda_t-\mu_t)  &	0	\\
				-\Im(\lambda_t+\mu_t)  &	 \Re(\lambda_t-\mu_t)  &	0	\\
				0				   &		0 			   & 	a_t+b_t-1	\\
			\end{pmatrix}.
		\end{equation*}
		The block diagonal structure of $\wt{\Lambda}(t)$ makes evident the orthogonal covariance of the map with respect to orthogonal transformations diagonal in the $\sigma_3$ basis. The generator can also be rewritten in GKSL Pauli form as in~\eqref{timedep_qubit} and~\eqref{timedependent_qubitgenerator}: with Hamiltonian
		\begin{equation}
			H(t)=\frac{\omega(t)}{2} \sigma_3 \ .
		\end{equation}
		Furthermore, letting $\kappa(t)=-\Re(m_t)$, $\eta(t)=-\Im(m_t)$ and $\delta(t)=(\gamma_+(t)-\gamma_-(t))/2$, with Kossakowski matrix
		\begin{equation*}
			K(t)=\begin{pmatrix}
				\frac{\Gamma_L(t)}{2}-\kappa(t) & \eta(t)+i\delta(t)& 0 \\
				\eta(t)-i\delta(t)  &  \frac{\Gamma_L(t)}{2}+\kappa(t) & 0\\
				0 & 0 & \Gamma_T(t)-\frac{\Gamma_L(t)}{2}
			\end{pmatrix}\ .
		\end{equation*}
		Notice that $CP$-divisibility conditions can be then read off asking for $K(t)\ge0$.
		If $\Lambda_t$ is also unital,  $A(t)$ has to be bistochastic, thus $a_t=b_t$ and
    $$
    \gamma_+(t)=\gamma_-(t)=\frac{\Gamma_L(t)}{2}=\frac{\dot a_t}{2a_t-1}\,,
    $$
  so that the $P$-divisibility conditions simplify to
	\begin{equation}\label{unital_PD}
		\Gamma_L(t)\ge 0 \, , \qquad \Gamma_T \ge \abs{m_t}\,.
	\end{equation}
%
%
%
%
	Indeed, the generator $\widetilde{\mathcal{L}}(t)$ will also be Bloch diagonal, and the conditions~\eqref{unital_PD} follow from $-(\wt{\mc{L}}(t)+\wt{\mc{L}}^T(t))\ge0$.
	On the other hand, the conditions for $CP$-divisibility in the unital case simplify to
	\begin{equation}\label{unital_CPD}
		\Gamma_T(t)\ge\frac{\Gamma_L(t)}{2}\ge \abs{m_t}\,.
	\end{equation}

	\end{remark}

\subsection{Non-self dual dynamics from self-dual generator.}
\label{sec:nonselfdualMAP_selfdualGEN}

As already stressed, generators of dynamical maps in the class~\eqref{class_notime} are also in the same class; moreover, we shall focus upon purely dissipative generator; for qubit unital dynamics, this amounts to requiring that $\mathcal{L}_t^{\phantom\ddagger}=\mathcal{L}_t^\ddagger$ (see Lemma \ref{lemma1}). From the general form \eqref{class_notime}, the dual of $\Phi^{(A,\lambda, \mu)}$ acts on a matrix $X\in \Md{2}$ as
\begin{align*}
	(\Phi^{(A,\lambda,\mu)})^\ddagger[X]=&\sum_{ij} A_{ji} E_{ij} X E_{ji} + \overline{\lambda} E_{00}X E_{11} \nonumber\\
	&\hskip-1cm+ \lambda E_{11}X E_{00}+ \mu E_{00}X^T E_{11}+ \overline{\mu} E_{11}X^T E_{00} \ .
\end{align*}
Then, self-duality corresponds to $A=A^T$ and $\lambda \in \mathbb{R}$.
Therefore, a generator $\mathcal{L}_t=\Phi^{(B(t),l_t, m_t)}$  is self-dual iff $\gamma_+(t)=\gamma_-(t)$, implying unitality of the dynamics, and $l_t$ is real:
\begin{equation}\label{selfdual_gen}
	l_t=-\Gamma_T(t)\,, \qquad \omega(t)=0\,.
\end{equation}
 Let $\lambda_t=\abs{\lambda_t}e^{i \varphi_t}$ and $\mu_t=\abs{\mu_t}e^{i \theta_t }$, $\varphi_t,\theta_t\in \mathbb{R}$;
then, \eqref{ell_emm} and \eqref{selfdual_gen} imply the following relation between the phases and the moduli of $\lambda_t$ and $\mu_t$ of ${\Lambda}_t=\Phi^{(A(t),\lambda_t, \mu_t)}$,
\begin{equation}\label{dissipativity}
\dot\varphi_t\abs{\lambda_t}^2=\dot \theta_t\abs{\mu_t}^2\,.
\end{equation}
Notice that this fixes $\dot\varphi_{t=0}=0$.
We now provide a construction for a $P$-divisible, non self-dual dynamics arising from a self-dual, purely dissipative generator. Let us define
\begin{equation}\label{gh}
	g_t:=\abs{\lambda_t}+\abs{\mu_t} \,, \qquad	h_t:=\abs{\lambda_t}-\abs{\mu_t}\,,
\end{equation}
with $0< h_t \le g_t$, $g_0=h_0=1$. Then,  \eqref{Rel_explicit} can be recast as
\begin{equation}\label{lt_construction}
 \Gamma_T(t)= -\frac{1}{2}\frac{\partial_t (g_t h_t)}{g_t h_t}=-\frac{\dot g_t}{2g_t}-\frac{\dot h_t}{2h_t}\,,
\end{equation}
which must be positive for $P$-divisibility.
On the other hand, substituting \eqref{gh} and the self-duality condition \eqref{dissipativity} into the second of  \eqref{ell_emm}, one recasts $\abs{m_t}$ as
\begin{align*}
	\abs{m_t}=\abs{\frac{\dot g_t}{2g_t}-\frac{\dot h_t}{2h_t}+ i \,\dot \theta_t\, \frac{g_t-h_t}{g_t+h_t}} \,.
\end{align*}
Taking the square of the inequality $\Gamma_T(t)\ge\abs{m_t}$,  the following necessary and sufficient conditions for $P$-divisibility can be finally found,
\begin{equation}\label{nec_suffcondition}
	|\dot\theta_t|^2 \left(\frac{g_t-h_t}{g_t+h_t}\right)^2\le \frac{\dot g_t }{g_t} \frac{\dot h_t}{h_t}\,.
\end{equation}
From \eqref{lt_construction} and \eqref{nec_suffcondition} one deduces, in particular, that $g_t$, $h_t$ must be monotonically decreasing.

\begin{example}
\label{newexample}
Let $C\ge0$ and define $\Lambda_t=\Phi^{(A(t),\lambda_t,\mu_t)}$ by
\begin{align}
	\abs{\lambda_t}&= e^{-2t} \cosh(t) \,, \qquad \varphi_t=C\tanh^3(t) \,,\\
	\abs{\mu_t}&= e^{-2t}\sinh(t) \,, \qquad  \ \theta_t=3 C \tanh(t)\,, \\
	a_t&=b_t=e^{-t}\,\,\cosh(t)  \,,
\end{align}
corresponding to the positive and monotonically decreasing functions $h_t=e^{-3t}$, $g_t=e^{-t}$. As one easily checks, \eqref{CP_class} are satisfied and the map is thus completely positive. On the other hand, the self-dual generator is given by
\begin{align}
	l_t=-\Gamma_T(t)&=-2\,, \qquad m_t=\sqrt{1+\, r_t^2}\,e^{i (\theta_t+\varphi_t)}\,,\\
	\gamma_+(t)&=\gamma_-(t)=\frac{\Gamma_L(t)}{2}=1 \,,
\end{align}
where $r_t=3C (1-\tanh^2(t))\tanh(t)$. Notice that
$$
\abs{m_t}=\sqrt{1+r_t^2}\ge1=\frac{\Gamma_L(t)}{2}\,,
$$ so the dynamics is $CP$-divisible iff $C=0$. On the other hand, the $P$-divisibility condition  \eqref{nec_suffcondition} reduces to
$r_t^2 \le 3$.
Since $r_t$ reaches a  maximum  of $2C/\sqrt{3}$, the dynamics is $P$-divisible iff $C\le 3/2$.
\end{example}

\subsection{Classical map from pure dissipation: loss of $P$-divisibility.}

For the class $\Lambda_t=\Phi^{(A(t),\lambda_t,\mu_t)}$, clearly, there exists a preferred basis $E_{00}, E_{11}$, yielding
\begin{equation}
    T(t)=A(t) \,,
\end{equation}
whose $P$-divisibility is necessary to ensure that of $\Lambda_t$ (see \eqref{nonunitalPD}). We shall now see that choosing a different basis generally breaks $P$-divisibility of $T(t)$ even for purely dissipative generators. Interestingly, for a suitable basis, this can occur with  $T(t)$  being invertible for all $t\ge0$.
A classical reduction of $\Lambda_t=\Phi^{(A(t),\lambda_t,\mu_t)}$ builds upon fixing a rank-1 projector $P_{\bs{n}}=\sum_{i,j=0,1}P_{\bs{n}}^{ij}E_{ij}$ and considering
\begin{align}\label{classicalmap_class}
		T_{00}(t)=\Tr(P_{\bs{n}}\Lambda_t[P_{\bs{n}}])=\sum_{ij} A_{ij}(t) & P_{\bs{n}}^{jj} P_{\bs{n}}^{ii} \nonumber\\\hskip-0.7cm +2\Re(\lambda_t) |P_{\bs{n}}^{01}|^2 + 2\Re(\mu_t (P_{\bs{n}}^{10})^2)\,.
\end{align}
which is sufficient to construct $T(t)$ in the unital case.
Letting $\bs{n}=(\sin(\chi)\cos(\xi),\sin(\chi)\sin(\xi),\cos(\chi))$, \eqref{classicalmap_class} can be then recast as
\begin{align}
	T_{00}(t)=&\frac{1}{2}\big(1+(2 a_t -1) \cos^2(\chi)+\abs{\lambda_t}\cos(\varphi_t)\sin^2(\chi)\nonumber\\
	&+\abs{\mu_t} \sin(\chi)\cos(\theta_t+2\xi)\big)\,.
\end{align}
As we already noted, for $\chi=0$, $P=E_{00}$, $T_{00}(t)=a_t$ yielding a $P$-divisible process. Considering instead ${\chi={\pi}/{2}}$, one has
\begin{align}
	T_{00}(t)=\frac{1}{2}\big(1+\abs{\lambda_t}\cos(\varphi_t)
	+\abs{\mu_t} \cos(\theta_t+2\xi)\big)\,.
\end{align}
Recalling that $P$-divisibility of the classical reduction amounts to check condition~\eqref{condition_classicalprocess0},
let $\xi=\pi/4$ and consider the map of Example \ref{newexample},
\begin{align}
	2T_{00}(t)-1=&\,e^{-2t}\cosh(t)\cos(C\tanh^3(t))\nonumber\\
	& -e^{-2t}\sinh(t) \sin(3C\tanh(t))\,.
\end{align}
Choosing $C=3/2$, \eqref{nec_suffcondition} saturates for some $t$. Noting that $0\le\varphi\le \pi/2$ and $0\le\theta\le3\pi/2$, so that $\sin(\theta_t)<0$ $\,\forall \,t>\rm{artanh}(2\pi/9)$, one can verify that $2 T_{00}(t)-1>0,\, \forall t\ge0$.
Thus, $T(t)$ is invertible. Nevertheless, as displayed in Fig.~\ref{fig:t0024}, it is not monotonically decreasing, since $\dot T_{00}(t)$ can become positive. Hence, $f_t$ also becomes positive implying that $T(t)$ is not $P$-divisible.
\begin{figure}
	\centering
	\includegraphics[width=0.9\linewidth]{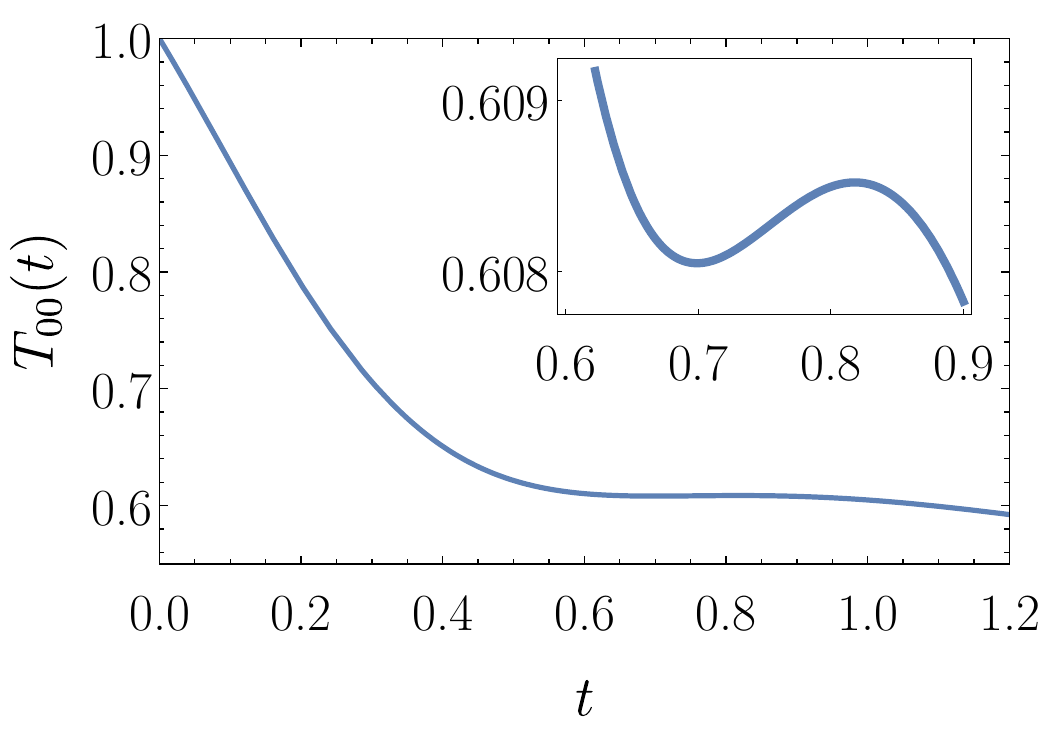}
	\caption{Plot of $T_{00}(t)$ from Example \ref{newexample}, with $C=3/2$ and $\xi=\pi/4$ corresponding to the Bloch vector $(\sqrt{2}/2,\sqrt{2}/2,0)$ which defines the reference classical basis.}
	\label{fig:t0024}
\end{figure}

\begin{remark}
	\label{rmk:nonPd}
Our previous considerations focussed on checking when $P$-divisibility of a purely dissipative qubit dynamical map is inherited by its classical reduction and found that in some cases it can be lost, namely the classical reduction can become non $P$-divisible. We now like to comment that the contrary can also occur; namely, the classical reduction of a non $P$-divisible purely dissipative qubit dynamics can become $P$-divisible. In Fig.~\ref{fig:bfquantum}, an instance of such behaviour is shown by the maps of Example~\ref{newexample} for a suitable $C>3/2$ for which $\Lambda_t$ is not $P$-divisible.
  Interestingly, the corresponding backflow of information can be witnessed by suitable orthogonal projections $P_{\pm\bs{n}}$ such that
	$$
	\dv{t}\frac{1}{2}\normt{\Lambda_t[P_{\bs{n}}-P_{-\bs{n}}]}=\dv{t}\|\wt{\Lambda}(t)\bs{n}\|>0\,,
	$$
    for some $t>0$, as shown by the non-monotonic behaviour of $\mathcal{I}^q_t(P_{\boldsymbol{n}},P_{-\boldsymbol{n}};1/2)=\|\wt{\Lambda}(t)\bs{n}\|$ in the inset of Fig.~\ref{fig:bfquantum}.
    Nevertheless, the stochastic process $T(t)$ defined through the subalgebra $\mathcal{P}_{\bs{n}}$ generated by them is $P$-divisible. This means that, unlike the quantum dynamics it originates from, such a $T(t)$ cannot exhibit classical backflow of information. The orthogonal projections $P_{\pm\bs{n}}$ provide classical probabilities and, concerning information backflow, they are concrete instances of the following behaviour.
  Suppose that there exist $\rho_{\bs{p}}, \rho_{\bs{q}}$ in a commutative subalgebra $\mathcal{P}$, that is classical probability distributions, and $t>s>0$ such that
	\begin{equation*}
	\mathcal{I}^q_t(\rho_{\boldsymbol{p}},\rho_{\boldsymbol{q}};\mu)-\mathcal{I}^q_s(\rho_{\boldsymbol{p}},\rho_{\boldsymbol{q}};\mu)>0\,.
	\end{equation*}
    In such case, $\Lambda_t$ is clearly not $P$-divisible since it shows backflow of information. Using~\eqref{balance_cl}, one now has
	\begin{equation}\label{balance_BFI}
		\mathcal{I}_t^{cl}(\boldsymbol{p},\boldsymbol{q};\mu)+
		\mathcal{C}_t(\boldsymbol{p},\boldsymbol{q};\mu) > \mathcal{I}_s^{cl}(\boldsymbol{p},\boldsymbol{q};\mu)+
		\mathcal{C}_s(\boldsymbol{p},\boldsymbol{q};\mu)\,.
	\end{equation}
	In addition, suppose that the classical dynamics $T(t)$, obtained by restricting the quantum dynamics $\Lambda_t$ on the same subalgebra $\mathcal{P}$, is $P$-divisible, yielding thus
	\begin{equation}
		\mc{C}_t(\bs{p},\bs{q};\mu)-\mc{C}_s(\bs{p},\bs{q};\mu)>\mathcal{I}_s^{cl}(\boldsymbol{p},\boldsymbol{q};\mu)-\mathcal{I}_t^{cl}(\boldsymbol{p},\boldsymbol{q};\mu) \ge 0 \,.
	\end{equation}
	Therefore, the backflow of information from the environment to the open system only affects the coherent contribution to $\mathcal{I}_{t}^q$ and cannot be witnessed by the classical reduced map.
\end{remark}
\begin{figure}
	\centering
    \includegraphics[width=0.95\linewidth]{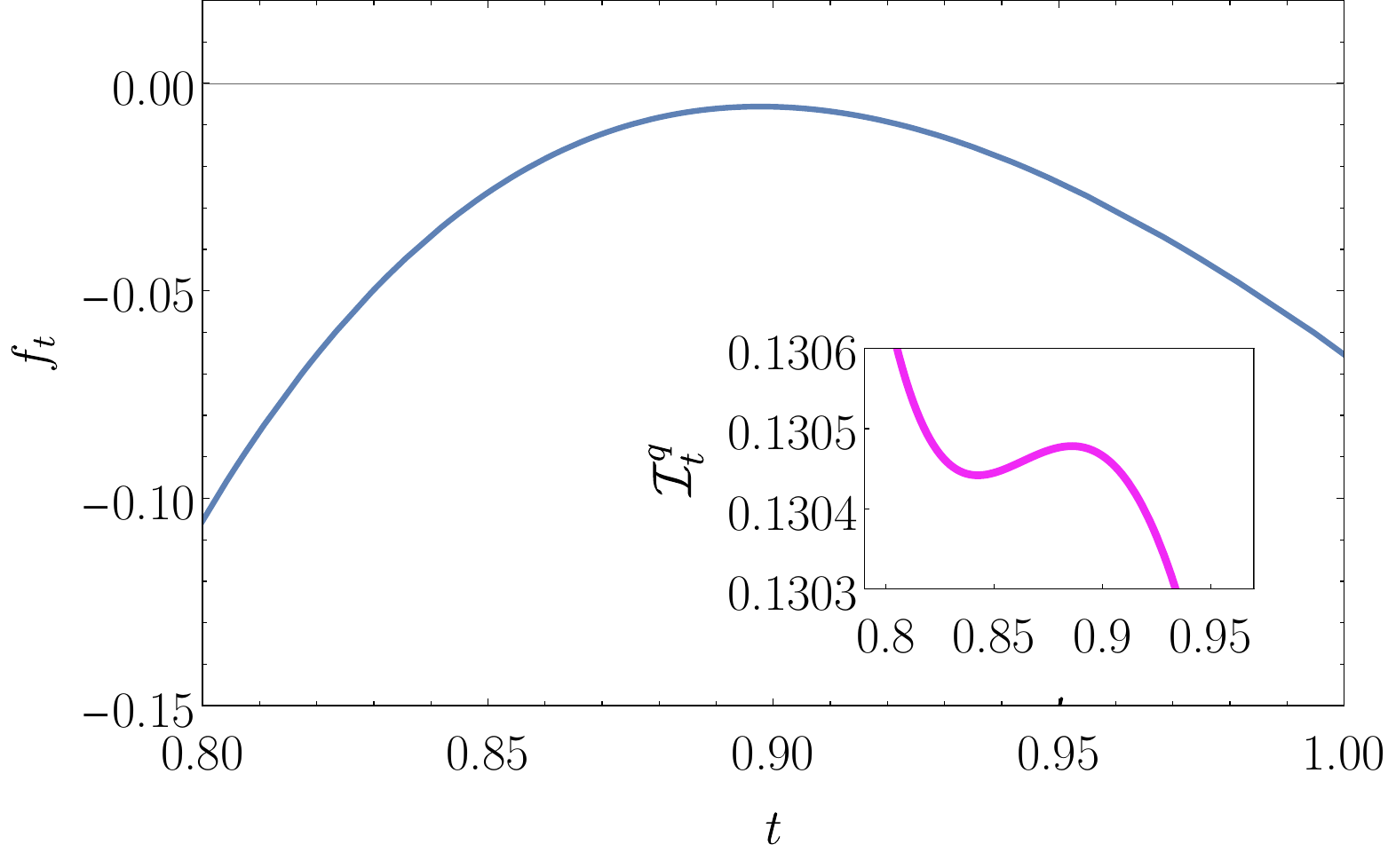}
	\caption{Global maximum of $f_t$ as defined in~\eqref{condition_classicalprocess0}, for the classical reduction of the map of Example~\ref{newexample}, with $C=1.64$ and reference classical algebra defined by the projector $P_{\boldsymbol{n}}$ with $\boldsymbol{n}=(\cos(\xi),\sin(\xi),0)$ and $\xi=\pi/8$. Numerically, it is checked that $f_t$ reaches a maximum of $-0.006$, so that $f_t<0$ for all $t$ and $T(t)$ is $P$-divisible. In the inset, the non-monotonic behaviour of $\mathcal{I}_t^q(P_{\boldsymbol{n}},P_{-\boldsymbol{n}}; 1/2)$
    is displayed, corresponding to a backflow of information that thus involves only the coherences w.r.t. such basis.}
	\label{fig:bfquantum}
\end{figure}

\section{Conclusions}

Generic quantum dynamics restricted to commutative algebras of orthogonal projections generate quantum coherences among them.
With respect to that commutative algebra, by restricting to the diagonal components of its time-evolving orthogonal projections, one obtains a so-called classical reduction of the quantum dynamics, of course depending on the chosen set of orthogonal projections.
In the manuscript, we have studied whether and how the divisibility properties of time-dependent, non-Markovian qubit quantum dynamics are inherited by their classical reductions. We have done it by classically reducing the dynamics and by then studying the generator of the classical reduction.

In Section~\ref{sec:BFI_interpretation}, we have argued that the lack of $P$-divisibility in classical reductions can be ascribed to information created by the quantum dynamics which is stored in the quantum coherences and then released back into the classical component of the dynamics.
Furthermore, for $P$-divisible and unital qubit dynamics we showed that 1) the Hamiltonian contributions to the time-dependent quantum generators generally give rise to qubit quantum dynamics with non $P$-divisible classical reductions; 2) purely dissipative, self-dual dynamics
always have $P$-divisible classical reductions while 3) purely dissipative, non self-dual dynamics may give rise to non $P$-divisible classical reductions.

Though this behaviour is somewhat typical of unitary quantum dynamics, yet it can also emerge from purely dissipative quantum evolutions due to the presence of a non-trivially time-ordered Dyson expansion.

\section*{Acknowledgments}
The authors heartily thank Roberto Floreanini for his comments on the paper. GN would also like to thank Frederik vom Ende for his remarks on the manuscript. FB and GN acknowledge financial support from the PNRR MUR project PE0000023-NQSTI. DC was supported by the Polish National Science Center project No. 2018/30/A/ST2/00837.

\appendix

\section{Proof of Proposition~\ref{propP-div}}
\label{App1}

\begin{proof}
	Let $P=\ketbra{\psi}$ be a generic $2\times 2$ projector, with $\ket{\psi}=(w_1,w_2)=(w_1, \abs{w_2}e^{-i \frac{\Omega}{2}})$,
$w_1,\Omega\in\mathbb{R}$, $\abs{w_1}^2+\abs{w_2}^2=1$  and let $Q$ be its orthogonal complement $Q=\mathds{1}_2-P$.
Letting also $m_t=\abs{m_t}e^{i\chi_t}$, one has
\begin{align}
&\Tr(Q\mathcal{L}_t[P])\nonumber\\
&=\sum_{ij} B_{ij}(t) Q_{ii}P_{jj}+ l_t P_{01} Q_{10} +\overline{l_t} P_{10} Q_{01} \nonumber \\
&\hskip1cm +m_t P_{10} Q_{10}+\overline{m_t} P_{01} Q_{01}\nonumber\\
&=\gamma_-(t)\abs{w_1}^4+\gamma_+(t)\abs{w_2}^4\nonumber+2(\Gamma_T(t)-\frac{\Gamma_L(t)}{2})\abs{w_1}^2\abs{w_2}^2\nonumber\\
&\hskip1cm	-2\abs{m_t}\abs{w_1}^2\abs{w_2}^2\cos(\chi_t-\Omega). \label{last}
\end{align}
For the sufficiency part, let $\gamma_{\pm}(t)\ge 0$ and
$$
\mathcal{G}(t):=\Gamma_T(t)-\frac{\Gamma_L(t)}{2}+\sqrt{\gamma_+(t)\gamma_-(t)}-\abs{m_t} \ge0\,.
$$
Then,
\begin{align*}\label{quadraticformPD}
&\Tr(Q\mathcal{L}_t[P])= \big(\sqrt{\gamma_-(t)}\abs{w_1}^2-\sqrt{\gamma_+(t)}\abs{w_2}^2\big)^2\\
&+2\big(\Gamma_t(t)-\frac{\Gamma_L(t)}{2}+\sqrt{\gamma_+(t)\gamma_-(t)} \big)\abs{w_1}^2\abs{w_2}^2\\&-2\,\abs{m_t}\cos(\chi_t-\Omega)\big)\abs{w_1}^2\abs{w_2}^2\\
&\ge	\big(\sqrt{\gamma_-(t)}\abs{w_1}^2-\sqrt{\gamma_+(t)}\abs{w_2}^2\big)^2+2\mathcal{G}(t)\abs{w_1}^2\abs{w_2}^2 \\
&\ge0 \,,
\end{align*}
so that
$P$-divisibility follows from Theorem~\ref{timekossakowskiconditions}.
To prove the necessity part, first notice that, asking for $\Tr(Q\mc{L}_t[P])\ge0$
 the choices $w_1=1$, and $w_2=1$ imply that $\gamma_{\pm}(t)\ge0$. Choose instead, for fixed $t\ge0$,
\begin{align*}
	\abs{w_1}^2=\frac{\sqrt{\gamma_+(t)}}{\sqrt{\gamma_+(t)}+\sqrt{\gamma_-(t)}} \,,  \ \abs{w_2}^2=\frac{\sqrt{\gamma_-(t)}}{\sqrt{\gamma_+(t)}+\sqrt{\gamma_-(t)}}\,,
\end{align*}
and $\Omega=\chi_t$ so that, from \eqref{last},
 \begin{align*}
 	0\le\Tr(Q\mathcal{L}_t[P])=2\mathcal{G}(t)\abs{w_1}^2\abs{w_2}^2,
 \end{align*}
 which implies \eqref{nontrivial_PD}. For $CP$-divisibility, a necessary and sufficient condition for $\mathcal{L}_t$ to be in the GKSL form for all $t\ge0$ is~\cite{EvansLewis77,Wolf2008Assessing}
\begin{equation}\label{CPdivisibility_condition}
	Y_t\equiv(\mathds{1}_4-P_2^+) \mathcal{L}_t\otimes\mathrm{id}_2[P_2^+] (\mathds{1}_4-P_2^+) \ge 0\,,
\end{equation}
Noting that
\begin{equation*}
	Y_t=\begin{pmatrix}
	\frac{2\Gamma_T(t)-\Gamma_L(t)}{4} & 0 & 0 & -\frac{2\Gamma_T(t)-\Gamma_L(t)}{4}\\
	0 & \gamma_+(t) &  m_t & 0\\
	0 &\overline{m_t} &  \gamma_-(t) & 0\\
	-\frac{2\Gamma_T(t)-\Gamma_L(t)}{4} & 0 & 0 & \frac{2\Gamma_T(t)-\Gamma_L(t)}{4}\\
\end{pmatrix}.
\end{equation*}
which is positive iff  conditions~\eqref{CPd1} and~\eqref{CPd2} are verified.
\end{proof}

\section{Proof of inequality \eqref{bound2_coh}}
\label{app2:bound}

From \eqref{balance_cl},
\begin{align*}
\Delta \mc{I}_{t,s}^{cl}(\boldsymbol{p},\boldsymbol{q};\mu)&=\mathcal{I}_t^{cl}(\boldsymbol{p},\boldsymbol{q};\mu)-\mathcal{I}_s^{cl}(\boldsymbol{p},\boldsymbol{q};\mu)\\
& \le
\mathcal{C}_s(\boldsymbol{p},\boldsymbol{q};\mu)-\mathcal{C}_t(\boldsymbol{p},\boldsymbol{q};\mu)\nonumber\\
& \le \mathcal{C}_s(\boldsymbol{p},\boldsymbol{q};\mu) \,,
\end{align*}
since $\mathcal{C}_s(\boldsymbol{p},\boldsymbol{q};\mu)\ge0$. Let then $\mathbb{P}^\perp:={\rm id}-\mathbb{P}$ and apply twice the triangle inequality, so to get
\begin{align}
  \mathcal{C}_s(\boldsymbol{p},\boldsymbol{q};\mu)&\leq \|\mathbb{P}^\perp\Lambda_s[\Delta_\mu(\boldsymbol{p},\boldsymbol{q})]\|_1 \nonumber\\
 &\le\,\mu\, \|\mathbb{P}^\perp\Lambda_s[\rho_{\boldsymbol{p}}]\|_1\,+\,(1-\mu)\,\|\mathbb{P}^\perp\Lambda_s[\rho_{\boldsymbol{q}}]\|_1\ .\label{coherencebound}
\end{align}
Noting that ${\normt{\ketbra{i}{j}}=1}$, with $\ket{i},\ket{j}$ belonging to the reference orthonormal basis, one applies the triangle inequality once again to the r.h.s. of~\eqref{coherencebound} so to obtain
\begin{align*}
     \Delta \mc{I}_{t,s}^{cl}(\boldsymbol{p},\boldsymbol{q};\mu) &
    \le\mu\, \|\mathbb{P}^\perp\Lambda_s[\rho_{\boldsymbol{p}}]\|_1\,+\,(1-\mu)\,\|\mathbb{P}^\perp\Lambda_s[\rho_{\boldsymbol{q}}]\|_1\ \nonumber \\
 &\le  \mu \ C_{\ell_1}(\Lambda_s[\rho_{\bs{p}}]) \, + \,(1-\mu)\ C_{\ell_1}(\Lambda_s[\rho_{\bs{q}}]) \,.
\end{align*}


\begin{thebibliography}{29}%
	\makeatletter
	\providecommand \@ifxundefined [1]{%
		\@ifx{#1\undefined}
	}%
	\providecommand \@ifnum [1]{%
		\ifnum #1\expandafter \@firstoftwo
		\else \expandafter \@secondoftwo
		\fi
	}%
	\providecommand \@ifx [1]{%
		\ifx #1\expandafter \@firstoftwo
		\else \expandafter \@secondoftwo
		\fi
	}%
	\providecommand \natexlab [1]{#1}%
	\providecommand \enquote  [1]{``#1''}%
	\providecommand \bibnamefont  [1]{#1}%
	\providecommand \bibfnamefont [1]{#1}%
	\providecommand \citenamefont [1]{#1}%
	\providecommand \href@noop [0]{\@secondoftwo}%
	\providecommand \href [0]{\begingroup \@sanitize@url \@href}%
	\providecommand \@href[1]{\@@startlink{#1}\@@href}%
	\providecommand \@@href[1]{\endgroup#1\@@endlink}%
	\providecommand \@sanitize@url [0]{\catcode `\\12\catcode `\$12\catcode
		`\&12\catcode `\#12\catcode `\^12\catcode `\_12\catcode `\%12\relax}%
	\providecommand \@@startlink[1]{}%
	\providecommand \@@endlink[0]{}%
	\providecommand \url  [0]{\begingroup\@sanitize@url \@url }%
	\providecommand \@url [1]{\endgroup\@href {#1}{\urlprefix }}%
	\providecommand \urlprefix  [0]{URL }%
	\providecommand \Eprint [0]{\href }%
	\providecommand \doibase [0]{https://doi.org/}%
	\providecommand \selectlanguage [0]{\@gobble}%
	\providecommand \bibinfo  [0]{\@secondoftwo}%
	\providecommand \bibfield  [0]{\@secondoftwo}%
	\providecommand \translation [1]{[#1]}%
	\providecommand \BibitemOpen [0]{}%
	\providecommand \bibitemStop [0]{}%
	\providecommand \bibitemNoStop [0]{.\EOS\space}%
	\providecommand \EOS [0]{\spacefactor3000\relax}%
	\providecommand \BibitemShut  [1]{\csname bibitem#1\endcsname}%
	\let\auto@bib@innerbib\@empty
	\bibitem [{\citenamefont {Alicki}\ and\ \citenamefont
		{Lendi}(2007)}]{AlickiLendi}%
	\BibitemOpen
	\bibfield  {author} {\bibinfo {author} {\bibfnamefont {R.}~\bibnamefont
			{Alicki}}\ and\ \bibinfo {author} {\bibfnamefont {K.}~\bibnamefont {Lendi}},\
	}\href {https://doi.org/https://doi.org/10.1007/3-540-70861-8} {\emph
		{\bibinfo {title} {Quantum dynamical semigroups and applications}}},\ Vol.\
	\bibinfo {volume} {717}\ (\bibinfo  {publisher} {Springer},\ \bibinfo {year}
	{2007})\BibitemShut {NoStop}%
	\bibitem [{\citenamefont {Breuer}\ and\ \citenamefont
		{Petruccione}(2002)}]{BreuerPetruccione}%
	\BibitemOpen
	\bibfield  {author} {\bibinfo {author} {\bibfnamefont {H.-P.}\ \bibnamefont
			{Breuer}}\ and\ \bibinfo {author} {\bibfnamefont {F.}~\bibnamefont
			{Petruccione}},\ }\href
	{https://doi.org/10.1093/acprof:oso/9780199213900.001.0001} {\emph {\bibinfo
			{title} {The {{Theory}} of {{Open Quantum Systems}}}}}\ (\bibinfo
	{publisher} {{Oxford University Press}},\ \bibinfo {year} {2002})\BibitemShut
	{NoStop}%
	\bibitem [{\citenamefont {Rivas}\ and\ \citenamefont
		{Huelga}(2012)}]{RivasHuelga}%
	\BibitemOpen
	\bibfield  {author} {\bibinfo {author} {\bibfnamefont {{\'A}.}~\bibnamefont
			{Rivas}}\ and\ \bibinfo {author} {\bibfnamefont {S.~F.}\ \bibnamefont
			{Huelga}},\ }\href {https://doi.org/10.1007/978-3-642-23354-8} {\emph
		{\bibinfo {title} {Open {{Quantum Systems}}: {{An Introduction}}}}},\
	{{SpringerBriefs}} in {{Physics}}\ (\bibinfo  {publisher} {{Springer}},\
	\bibinfo {address} {{Berlin, Heidelberg}},\ \bibinfo {year}
	{2012})\BibitemShut {NoStop}%
	\bibitem [{\citenamefont {Gorini}\ \emph {et~al.}(1976)\citenamefont {Gorini},
		\citenamefont {Kossakowski},\ and\ \citenamefont
		{Sudarshan}}]{GoriniKossSud}%
	\BibitemOpen
	\bibfield  {author} {\bibinfo {author} {\bibfnamefont {V.}~\bibnamefont
			{Gorini}}, \bibinfo {author} {\bibfnamefont {A.}~\bibnamefont
			{Kossakowski}},\ and\ \bibinfo {author} {\bibfnamefont {E.~C.~G.}\
			\bibnamefont {Sudarshan}},\ }\bibfield  {title} {\bibinfo {title} {Completely
			positive dynamical semigroups of {N}‐level systems},\ }\href
	{https://doi.org/https://doi.org/10.1063/1.522979} {\bibfield  {journal}
		{\bibinfo  {journal} {Journal of Mathematical Physics}\ }\textbf {\bibinfo
			{volume} {17}},\ \bibinfo {pages} {821} (\bibinfo {year} {1976})}\BibitemShut
	{NoStop}%
	\bibitem [{\citenamefont {Lindblad}(1976)}]{LindbladTh}%
	\BibitemOpen
	\bibfield  {author} {\bibinfo {author} {\bibfnamefont {G.}~\bibnamefont
			{Lindblad}},\ }\bibfield  {title} {\bibinfo {title} {On the generators of
			quantum dynamical semigroups},\ }\href
	{https://doi.org/https://doi.org/10.1007/BF01608499} {\bibfield  {journal}
		{\bibinfo  {journal} {Commun. Math. Phys.}\ }\textbf {\bibinfo {volume}
			{48}},\ \bibinfo {pages} {119} (\bibinfo {year} {1976})}\BibitemShut
	{NoStop}%
	\bibitem [{\citenamefont {Chru{\'s}ci{\'n}ski}(2022)}]{ChrusReview22}%
	\BibitemOpen
	\bibfield  {author} {\bibinfo {author} {\bibfnamefont {D.}~\bibnamefont
			{Chru{\'s}ci{\'n}ski}},\ }\bibfield  {title} {\bibinfo {title} {Dynamical
			maps beyond markovian regime},\ }\href
	{https://www.sciencedirect.com/science/article/pii/S0370157322003428}
	{\bibfield  {journal} {\bibinfo  {journal} {Physics Reports}\ }\textbf
		{\bibinfo {volume} {992}},\ \bibinfo {pages} {1} (\bibinfo {year}
		{2022})}\BibitemShut {NoStop}%
	\bibitem [{\citenamefont {Breuer}\ \emph {et~al.}(2009)\citenamefont {Breuer},
		\citenamefont {Laine},\ and\ \citenamefont {Piilo}}]{BLP}%
	\BibitemOpen
	\bibfield  {author} {\bibinfo {author} {\bibfnamefont {H.-P.}\ \bibnamefont
			{Breuer}}, \bibinfo {author} {\bibfnamefont {E.-M.}\ \bibnamefont {Laine}},\
		and\ \bibinfo {author} {\bibfnamefont {J.}~\bibnamefont {Piilo}},\ }\bibfield
	{title} {\bibinfo {title} {Measure for the degree of non-{M}arkovian
			behavior of quantum processes in {O}pen {S}ystems},\ }\href
	{https://link.aps.org/doi/10.1103/PhysRevLett.103.210401} {\bibfield
		{journal} {\bibinfo  {journal} {Phys. Rev. Lett.}\ }\textbf {\bibinfo
			{volume} {103}},\ \bibinfo {pages} {210401} (\bibinfo {year}
		{2009})}\BibitemShut {NoStop}%
	\bibitem [{\citenamefont {Breuer}\ \emph {et~al.}(2016)\citenamefont {Breuer},
		\citenamefont {Laine}, \citenamefont {Piilo},\ and\ \citenamefont
		{Vacchini}}]{ColloquiumBreuer}%
	\BibitemOpen
	\bibfield  {author} {\bibinfo {author} {\bibfnamefont {H.-P.}\ \bibnamefont
			{Breuer}}, \bibinfo {author} {\bibfnamefont {E.-M.}\ \bibnamefont {Laine}},
		\bibinfo {author} {\bibfnamefont {J.}~\bibnamefont {Piilo}},\ and\ \bibinfo
		{author} {\bibfnamefont {B.}~\bibnamefont {Vacchini}},\ }\bibfield  {title}
	{\bibinfo {title} {Colloquium: non-{M}arkovian dynamics in open quantum
			systems},\ }\href {https://link.aps.org/doi/10.1103/RevModPhys.88.021002}
	{\bibfield  {journal} {\bibinfo  {journal} {Rev. Mod. Phys.}\ }\textbf
		{\bibinfo {volume} {88}},\ \bibinfo {pages} {021002} (\bibinfo {year}
		{2016})}\BibitemShut {NoStop}%
	\bibitem [{\citenamefont {Benatti}\ and\ \citenamefont
		{Nichele}(2023)}]{FBGN2023}%
	\BibitemOpen
	\bibfield  {author} {\bibinfo {author} {\bibfnamefont {F.}~\bibnamefont
			{Benatti}}\ and\ \bibinfo {author} {\bibfnamefont {G.}~\bibnamefont
			{Nichele}},\ }\bibfield  {title} {\bibinfo {title} {{O}pen {Q}uantum
			{D}ynamics: {M}emory {E}ffects and {S}uperactivation of {B}ackflow of
			{I}nformation},\ }\href
	{https://doi.org/https://doi.org/10.3390/math12010037} {\bibfield  {journal}
		{\bibinfo  {journal} {Mathematics}\ }\textbf {\bibinfo {volume} {12}},\
		\bibinfo {pages} {37} (\bibinfo {year} {2023})}\BibitemShut {NoStop}%
	\bibitem [{\citenamefont {Korzekwa}\ and\ \citenamefont
		{Lostaglio}(2021)}]{LostaglioKorzekwa2021quantum}%
	\BibitemOpen
	\bibfield  {author} {\bibinfo {author} {\bibfnamefont {K.}~\bibnamefont
			{Korzekwa}}\ and\ \bibinfo {author} {\bibfnamefont {M.}~\bibnamefont
			{Lostaglio}},\ }\bibfield  {title} {\bibinfo {title} {Quantum advantage in
			simulating stochastic processes},\ }\href
	{https://doi.org/10.1103/PhysRevX.11.021019} {\bibfield  {journal} {\bibinfo
			{journal} {Physical Review X}\ }\textbf {\bibinfo {volume} {11}},\ \bibinfo
		{pages} {021019} (\bibinfo {year} {2021})}\BibitemShut {NoStop}%
	\bibitem [{\citenamefont {Shahbeigi}\ \emph {et~al.}(2023)\citenamefont
		{Shahbeigi}, \citenamefont {Chubb}, \citenamefont {Kukulski}, \citenamefont
		{Pawela},\ and\ \citenamefont {Korzekwa}}]{shahbeigi2023quantum}%
	\BibitemOpen
	\bibfield  {author} {\bibinfo {author} {\bibfnamefont {F.}~\bibnamefont
			{Shahbeigi}}, \bibinfo {author} {\bibfnamefont {C.~T.}\ \bibnamefont
			{Chubb}}, \bibinfo {author} {\bibfnamefont {R.}~\bibnamefont {Kukulski}},
		\bibinfo {author} {\bibfnamefont {{\L}.}~\bibnamefont {Pawela}},\ and\
		\bibinfo {author} {\bibfnamefont {K.}~\bibnamefont {Korzekwa}},\ }\bibfield
	{title} {\bibinfo {title} {Quantum-embeddable stochastic matrices},\ }\href
	{https://doi.org/10.48550/arXiv.2305.17163} {\bibfield  {journal} {\bibinfo
			{journal} {arXiv preprint arXiv:2305.17163}\ } (\bibinfo {year}
		{2023})}\BibitemShut {NoStop}%
	\bibitem [{\citenamefont {Kossakowski}(1972)}]{Kossakowski}%
	\BibitemOpen
	\bibfield  {author} {\bibinfo {author} {\bibfnamefont {A.}~\bibnamefont
			{Kossakowski}},\ }\bibfield  {title} {\bibinfo {title} {On necessary and
			sufficient conditions for a generator of a quantum dynamical semi-group},\
	}\href@noop {} {\bibfield  {journal} {\bibinfo  {journal} {Bull. Acad. Pol.
				Sci. Ser. Math. Astr. Phys.}\ }\textbf {\bibinfo {volume} {20}},\ \bibinfo
		{pages} {1021} (\bibinfo {year} {1972})}\BibitemShut {NoStop}%
	\bibitem [{\citenamefont {Van~Kampen}(1992)}]{VanKampen}%
	\BibitemOpen
	\bibfield  {author} {\bibinfo {author} {\bibfnamefont {N.~G.}\ \bibnamefont
			{Van~Kampen}},\ }\href
	{https://doi.org/https://doi.org/10.1016/B978-0-444-52965-7.X5000-4} {\emph
		{\bibinfo {title} {Stochastic processes in physics and chemistry}}},\
	Vol.~\bibinfo {volume} {1}\ (\bibinfo  {publisher} {Elsevier},\ \bibinfo
	{year} {1992})\BibitemShut {NoStop}%
	\bibitem [{\citenamefont {St{\o}rmer}(2013)}]{PaulsenPositive}%
	\BibitemOpen
	\bibfield  {author} {\bibinfo {author} {\bibfnamefont {E.}~\bibnamefont
			{St{\o}rmer}},\ }\href
	{https://doi.org/https://doi.org/10.1007/978-3-642-34369-8} {\emph {\bibinfo
			{title} {Positive {L}inear {M}aps of {O}perator {A}lgebras}}}\ (\bibinfo
	{publisher} {Springer},\ \bibinfo {year} {2013})\BibitemShut {NoStop}%
	\bibitem [{\citenamefont {Milz}\ \emph {et~al.}(2020)\citenamefont {Milz},
		\citenamefont {Egloff}, \citenamefont {Taranto}, \citenamefont {Theurer},
		\citenamefont {Plenio}, \citenamefont {Smirne},\ and\ \citenamefont
		{Huelga}}]{milz2020non}%
	\BibitemOpen
	\bibfield  {author} {\bibinfo {author} {\bibfnamefont {S.}~\bibnamefont
			{Milz}}, \bibinfo {author} {\bibfnamefont {D.}~\bibnamefont {Egloff}},
		\bibinfo {author} {\bibfnamefont {P.}~\bibnamefont {Taranto}}, \bibinfo
		{author} {\bibfnamefont {T.}~\bibnamefont {Theurer}}, \bibinfo {author}
		{\bibfnamefont {M.~B.}\ \bibnamefont {Plenio}}, \bibinfo {author}
		{\bibfnamefont {A.}~\bibnamefont {Smirne}},\ and\ \bibinfo {author}
		{\bibfnamefont {S.~F.}\ \bibnamefont {Huelga}},\ }\bibfield  {title}
	{\bibinfo {title} {When is a non-{M}arkovian quantum process classical?},\
	}\href {https://doi.org/10.1103/PhysRevX.10.041049} {\bibfield  {journal}
		{\bibinfo  {journal} {Physical Review X}\ }\textbf {\bibinfo {volume} {10}},\
		\bibinfo {pages} {041049} (\bibinfo {year} {2020})}\BibitemShut {NoStop}%
	\bibitem [{\citenamefont {Chru{\'s}ci{\'n}ski}\ \emph
		{et~al.}(2011)\citenamefont {Chru{\'s}ci{\'n}ski}, \citenamefont
		{Kossakowski},\ and\ \citenamefont {Rivas}}]{CKR}%
	\BibitemOpen
	\bibfield  {author} {\bibinfo {author} {\bibfnamefont {D.}~\bibnamefont
			{Chru{\'s}ci{\'n}ski}}, \bibinfo {author} {\bibfnamefont {A.}~\bibnamefont
			{Kossakowski}},\ and\ \bibinfo {author} {\bibfnamefont {{\'A}.}~\bibnamefont
			{Rivas}},\ }\bibfield  {title} {\bibinfo {title} {Measures of
			non-{{Markovianity}}: {{Divisibility}} versus backflow of information},\
	}\href {https://doi.org/10.1103/PhysRevA.83.052128} {\bibfield  {journal}
		{\bibinfo  {journal} {Phys. Rev. A}\ }\textbf {\bibinfo {volume} {83}},\
		\bibinfo {pages} {052128} (\bibinfo {year} {2011})}\BibitemShut {NoStop}%
	\bibitem [{\citenamefont {Rivas}(2022)}]{Rivas_counterexample}%
	\BibitemOpen
	\bibfield  {author} {\bibinfo {author} {\bibfnamefont {{\'A}.}~\bibnamefont
			{Rivas}},\ }\bibfield  {title} {\bibinfo {title} {Unidirectional information
			flow and positive divisibility are nonequivalent notions of quantum
			markovianity for noninvertible dynamics},\ }\href
	{https://doi.org/https://doi.org/10.1142/S1230161222500123} {\bibfield
		{journal} {\bibinfo  {journal} {Open Systems \& Information Dynamics}\
		}\textbf {\bibinfo {volume} {29}},\ \bibinfo {pages} {2250012} (\bibinfo
		{year} {2022})}\BibitemShut {NoStop}%
	\bibitem [{\citenamefont {Amato}\ \emph {et~al.}(2018)\citenamefont {Amato},
		\citenamefont {Breuer},\ and\ \citenamefont {Vacchini}}]{GTD_bound}%
	\BibitemOpen
	\bibfield  {author} {\bibinfo {author} {\bibfnamefont {G.}~\bibnamefont
			{Amato}}, \bibinfo {author} {\bibfnamefont {H.-P.}\ \bibnamefont {Breuer}},\
		and\ \bibinfo {author} {\bibfnamefont {B.}~\bibnamefont {Vacchini}},\
	}\bibfield  {title} {\bibinfo {title} {Generalized trace distance approach to
			quantum non-markovianity and detection of initial correlations},\ }\href
	{https://doi.org/10.1103/PhysRevA.98.012120} {\bibfield  {journal} {\bibinfo
			{journal} {Phys. Rev. A}\ }\textbf {\bibinfo {volume} {98}},\ \bibinfo
		{pages} {012120} (\bibinfo {year} {2018})}\BibitemShut {NoStop}%
	\bibitem [{\citenamefont {Smirne}\ \emph {et~al.}(2022)\citenamefont {Smirne},
		\citenamefont {Megier},\ and\ \citenamefont {Vacchini}}]{HolevoSkew}%
	\BibitemOpen
	\bibfield  {author} {\bibinfo {author} {\bibfnamefont {A.}~\bibnamefont
			{Smirne}}, \bibinfo {author} {\bibfnamefont {N.}~\bibnamefont {Megier}},\
		and\ \bibinfo {author} {\bibfnamefont {B.}~\bibnamefont {Vacchini}},\
	}\bibfield  {title} {\bibinfo {title} {Holevo skew divergence for the
			characterization of information backflow},\ }\href
	{https://doi.org/10.1103/PhysRevA.106.012205} {\bibfield  {journal} {\bibinfo
			{journal} {Phys. Rev. A}\ }\textbf {\bibinfo {volume} {106}},\ \bibinfo
		{pages} {012205} (\bibinfo {year} {2022})}\BibitemShut {NoStop}%
	\bibitem [{\citenamefont {Baumgratz}\ \emph {et~al.}(2014)\citenamefont
		{Baumgratz}, \citenamefont {Cramer},\ and\ \citenamefont
		{Plenio}}]{QuantifyingCoherence}%
	\BibitemOpen
	\bibfield  {author} {\bibinfo {author} {\bibfnamefont {T.}~\bibnamefont
			{Baumgratz}}, \bibinfo {author} {\bibfnamefont {M.}~\bibnamefont {Cramer}},\
		and\ \bibinfo {author} {\bibfnamefont {M.~B.}\ \bibnamefont {Plenio}},\
	}\bibfield  {title} {\bibinfo {title} {Quantifying coherence},\ }\href
	{https://doi.org/10.1103/PhysRevLett.113.140401} {\bibfield  {journal}
		{\bibinfo  {journal} {Phys. Rev. Lett.}\ }\textbf {\bibinfo {volume} {113}},\
		\bibinfo {pages} {140401} (\bibinfo {year} {2014})}\BibitemShut {NoStop}%
	\bibitem [{\citenamefont {Streltsov}\ \emph {et~al.}(2017)\citenamefont
		{Streltsov}, \citenamefont {Adesso},\ and\ \citenamefont
		{Plenio}}]{ColloquiumCoherence}%
	\BibitemOpen
	\bibfield  {author} {\bibinfo {author} {\bibfnamefont {A.}~\bibnamefont
			{Streltsov}}, \bibinfo {author} {\bibfnamefont {G.}~\bibnamefont {Adesso}},\
		and\ \bibinfo {author} {\bibfnamefont {M.~B.}\ \bibnamefont {Plenio}},\
	}\bibfield  {title} {\bibinfo {title} {Colloquium: Quantum coherence as a
			resource},\ }\href {https://doi.org/10.1103/RevModPhys.89.041003} {\bibfield
		{journal} {\bibinfo  {journal} {Reviews of Modern Physics}\ }\textbf
		{\bibinfo {volume} {89}},\ \bibinfo {pages} {041003} (\bibinfo {year}
		{2017})}\BibitemShut {NoStop}%
	\bibitem [{\citenamefont {Singh}\ and\ \citenamefont
		{Nechita}(2021)}]{SinghNecita2021}%
	\BibitemOpen
	\bibfield  {author} {\bibinfo {author} {\bibfnamefont {S.}~\bibnamefont
			{Singh}}\ and\ \bibinfo {author} {\bibfnamefont {I.}~\bibnamefont
			{Nechita}},\ }\bibfield  {title} {\bibinfo {title} {Diagonal unitary and
			orthogonal symmetries in quantum theory},\ }\href
	{https://doi.org/10.22331/q-2021-08-09-519} {\bibfield  {journal} {\bibinfo
			{journal} {Quantum}\ }\textbf {\bibinfo {volume} {5}},\ \bibinfo {pages}
		{519} (\bibinfo {year} {2021})},\ \Eprint {https://arxiv.org/abs/2010.07898}
	{2010.07898} \BibitemShut {NoStop}%
	\bibitem [{\citenamefont {Chru{\'s}ci{\'n}ski}\ and\ \citenamefont
		{Bhattacharya}(2024)}]{Bihalan}%
	\BibitemOpen
	\bibfield  {author} {\bibinfo {author} {\bibfnamefont {D.}~\bibnamefont
			{Chru{\'s}ci{\'n}ski}}\ and\ \bibinfo {author} {\bibfnamefont
			{B.}~\bibnamefont {Bhattacharya}},\ }\bibfield  {title} {\bibinfo {title} {A
			class of {S}chwarz qubit maps with diagonal unitary and orthogonal
			symmetries},\ }\href {https://doi.org/10.1088/1751-8121/ad75d6} {\bibfield
		{journal} {\bibinfo  {journal} {J. Phys. A: Math. Theor.}\ }\textbf {\bibinfo
			{volume} {57}},\ \bibinfo {pages} {395202} (\bibinfo {year}
		{2024})}\BibitemShut {NoStop}%
	\bibitem [{\citenamefont {Hall}(2008)}]{HallBlocks}%
	\BibitemOpen
	\bibfield  {author} {\bibinfo {author} {\bibfnamefont {M.~J.}\ \bibnamefont
			{Hall}},\ }\bibfield  {title} {\bibinfo {title} {Complete positivity for
			time-dependent qubit master equations},\ }\href
	{https://doi.org/10.1088/1751-8121/41/26/269801} {\bibfield  {journal}
		{\bibinfo  {journal} {J. Phys. A: Math. Theor.}\ }\textbf {\bibinfo {volume}
			{41}},\ \bibinfo {pages} {205302} (\bibinfo {year} {2008})}\BibitemShut
	{NoStop}%
	\bibitem [{\citenamefont {Cabrera}\ \emph {et~al.}(2019)\citenamefont
		{Cabrera}, \citenamefont {Davalos},\ and\ \citenamefont
		{Gorin}}]{CabreraDavalosGorin2019}%
	\BibitemOpen
	\bibfield  {author} {\bibinfo {author} {\bibfnamefont {G.~M.}\ \bibnamefont
			{Cabrera}}, \bibinfo {author} {\bibfnamefont {D.}~\bibnamefont {Davalos}},\
		and\ \bibinfo {author} {\bibfnamefont {T.}~\bibnamefont {Gorin}},\ }\bibfield
	{title} {\bibinfo {title} {Positivity and complete positivity of
			differentiable quantum processes},\ }\href
	{https://doi.org/https://doi.org/10.1016/j.physleta.2019.05.049} {\bibfield
		{journal} {\bibinfo  {journal} {Physics Letters A}\ }\textbf {\bibinfo
			{volume} {383}},\ \bibinfo {pages} {2719} (\bibinfo {year}
		{2019})}\BibitemShut {NoStop}%
	\bibitem [{\citenamefont {Kimura}(2002)}]{KimuraQubit}%
	\BibitemOpen
	\bibfield  {author} {\bibinfo {author} {\bibfnamefont {G.}~\bibnamefont
			{Kimura}},\ }\bibfield  {title} {\bibinfo {title} {Restriction on relaxation
			times derived from the lindblad-type master equations for two-level
			systems},\ }\href {https://doi.org/10.1103/PhysRevA.66.062113} {\bibfield
		{journal} {\bibinfo  {journal} {Physical Review A}\ }\textbf {\bibinfo
			{volume} {66}},\ \bibinfo {pages} {062113} (\bibinfo {year}
		{2002})}\BibitemShut {NoStop}%
	\bibitem [{\citenamefont {Chru{\'s}ci{\'n}ski}\ \emph
		{et~al.}(2021)\citenamefont {Chru{\'s}ci{\'n}ski}, \citenamefont {Kimura},
		\citenamefont {Kossakowski},\ and\ \citenamefont
		{Shishido}}]{chruscinski2021universal}%
	\BibitemOpen
	\bibfield  {author} {\bibinfo {author} {\bibfnamefont {D.}~\bibnamefont
			{Chru{\'s}ci{\'n}ski}}, \bibinfo {author} {\bibfnamefont {G.}~\bibnamefont
			{Kimura}}, \bibinfo {author} {\bibfnamefont {A.}~\bibnamefont
			{Kossakowski}},\ and\ \bibinfo {author} {\bibfnamefont {Y.}~\bibnamefont
			{Shishido}},\ }\bibfield  {title} {\bibinfo {title} {Universal constraint for
			relaxation rates for quantum dynamical semigroup},\ }\href
	{https://doi.org/10.1103/PhysRevLett.127.050401} {\bibfield  {journal}
		{\bibinfo  {journal} {Physical Review Letters}\ }\textbf {\bibinfo {volume}
			{127}},\ \bibinfo {pages} {050401} (\bibinfo {year} {2021})}\BibitemShut
	{NoStop}%
	\bibitem [{\citenamefont {Evans}\ and\ \citenamefont
		{Lewis}(1977)}]{EvansLewis77}%
	\BibitemOpen
	\bibfield  {author} {\bibinfo {author} {\bibfnamefont {D.~E.}\ \bibnamefont
			{Evans}}\ and\ \bibinfo {author} {\bibfnamefont {J.~T.}\ \bibnamefont
			{Lewis}},\ }\href
	{https://orca.cardiff.ac.uk/id/eprint/34031/1/Evans_Lewis_DIAS.pdf} {\emph
		{\bibinfo {title} {Dilations of irreversible evolutions in algebraic quantum
				theory}}},\ \bibinfo {number} {24}\ (\bibinfo  {publisher} {Dublin Institute
		for Advanced Studies},\ \bibinfo {year} {1977})\BibitemShut {NoStop}%
	\bibitem [{\citenamefont {Wolf}\ \emph {et~al.}(2008)\citenamefont {Wolf},
		\citenamefont {Eisert}, \citenamefont {Cubitt},\ and\ \citenamefont
		{Cirac}}]{Wolf2008Assessing}%
	\BibitemOpen
	\bibfield  {author} {\bibinfo {author} {\bibfnamefont {M.~M.}\ \bibnamefont
			{Wolf}}, \bibinfo {author} {\bibfnamefont {J.}~\bibnamefont {Eisert}},
		\bibinfo {author} {\bibfnamefont {T.~S.}\ \bibnamefont {Cubitt}},\ and\
		\bibinfo {author} {\bibfnamefont {J.~I.}\ \bibnamefont {Cirac}},\ }\bibfield
	{title} {\bibinfo {title} {Assessing {{Non-Markovian Quantum Dynamics}}},\
	}\href {https://doi.org/10.1103/PhysRevLett.101.150402} {\bibfield  {journal}
		{\bibinfo  {journal} {Phys. Rev. Lett.}\ }\textbf {\bibinfo {volume} {101}},\
		\bibinfo {pages} {150402} (\bibinfo {year} {2008})}\BibitemShut {NoStop}%
\end{thebibliography}
\end{document}